\documentclass[conference]{IEEEtran}
\IEEEoverridecommandlockouts
\usepackage{amsmath,amssymb,amsthm}
\usepackage{algorithm}
\usepackage{algorithmic}
\usepackage{graphicx}
\usepackage[caption=false,font=footnotesize]{subfig}
\usepackage{cite}
\usepackage{url}
\usepackage{booktabs}
\newtheorem{theorem}{Theorem}
\newtheorem{lemma}{Lemma}
\newtheorem{proposition}{Proposition}
\newtheorem{corollary}{Corollary}
\newtheorem{definition}{Definition}
\newtheorem{assumption}{Assumption}
\newtheorem{remark}{Remark}

\title{Delay-Adaptive Speculation Control for Low-Latency Edge-Cloud LLM Inference}

\author{
\IEEEauthorblockN{\parbox[t]{0.96\linewidth}{\centering
Kangkang Sun, \IEEEmembership{Member,~IEEE}, Jianhua Li, \IEEEmembership{Senior Member,~IEEE}, Xiuzhen Chen, \IEEEmembership{Member,~IEEE}, Junyi He, \IEEEmembership{Member,~IEEE}, and Minyi Guo, \IEEEmembership{Fellow,~IEEE}}}
\thanks{Kangkang Sun, Jianhua Li, Xiuzhen Chen and Minyi Guo are
with the Shanghai Key Laboratory of Integrated Administration Technologies for Information Security, School of Computer Science, Shanghai Jiao Tong University, Shanghai 200240, China (e-mail: szpsunkk@sjtu.edu.cn; lijh888@sjtu.edu.cn; chenxz@sjtu.edu.cn; guo-my@cs.sjtu.edu.cn).}
\thanks{Junyi He is  School of Science and Engineering, The Chinese University of Hong Kong, Shenzhen, Guangdong, China (e-mail: junyihe@link.cuhk.edu.cn).}
\thanks{\textit{The corresponding author is Jianhua Li.}}

\newcommand{\submissionjournal}{an IEEE journal}
\newcommand{\arxivurl}{}
\thanks{This manuscript has been submitted to \submissionjournal\ for possible publication. Copyright may be transferred without notice, after which this version may no longer be accessible.}%
}

\begin{document}
\maketitle

\begin{abstract}
Speculative decoding can accelerate large language model (LLM) inference by
letting a lightweight draft model generate candidate tokens that a larger target
model verifies in parallel. In distributed edge-cloud inference, however,
selecting the draft length becomes a challenging \emph{runtime control}
problem: longer drafts amortize communication delay but suffer decaying token
acceptance, whereas shorter drafts maintain higher acceptance at the cost of more
frequent communication rounds. We formulate this tradeoff as a ratio-type
optimal stopping problem and show that the optimal draft length is a finite
delay-monotone threshold. Our analysis further reveals a sharp phase transition
at a critical delay below which single-token speculation is always optimal, and
establishes that the optimal draft length grows only logarithmically with
communication delay. For time-varying networks, we extend the formulation to
Markov-modulated channels and show that, under a bounded speculation horizon
$K_{\max}$ and monotone stopping-region conditions, the optimal policy is a
state-dependent threshold. For unknown environments, we propose UCB-SpecStop,
an online control algorithm; with a sufficiently large exploration parameter
$\beta$, it satisfies gap-free and gap-dependent \emph{expected} regret bounds of order
$O\!\bigl(L_{\max}\sqrt{K_{\max} T \log(K_{\max} T)}\bigr)$ and
$O\!\bigl(\sum_{k:\Delta_k>0} L_{\max}^2\log(K_{\max} T)/\Delta_k\bigr)$, respectively.
We implement and evaluate the approach on a real edge-cloud testbed using an
NVIDIA Jetson Orin Nano Super as the edge node and an RTX~3090 Ti as the cloud
node with Qwen and Llama draft--target pairs. Experiments confirm the
predicted qualitative phase-transition behavior, with measured transition
points around 83~ms and 111~ms for the two model suites. The Qwen transition
closely matches the geometric prediction, while the LLaMA transition is
better explained after empirical-prefix calibration due to its heavy-head
acceptance profile. Across the tested delay grid, UCB-SpecStop reduces per-token latency over SpecDec++ by up to 22.4\% (Qwen at 20\,ms); in communication-dominated regimes, it approaches an offline best-fixed-arm empirical oracle within 0.2--2.4\%.
It also improves over a naive UCB baseline by up to 7.5\%,
removes the 14.0--18.7\% performance gap incurred by static draft-length tuning
under delay drift, while a contextual extension exploiting channel-state
information provides an additional 3.0--6.8\% gain.
\end{abstract}


\section{Introduction}
\label{sec:intro}

Large language models (LLMs) are increasingly being deployed in
latency-sensitive edge applications such as mobile assistants, on-device
agents, and interactive inference services, where response delay directly
affects user experience \cite{qu2025mobile,zeng2025h2o, xu2025joint}. Speculative
decoding~\cite{leviathan2023fast,chen2023accelerating} has emerged as an
effective way to reduce decoding latency by allowing a lightweight
\textit{draft model} to propose multiple tokens that a larger \textit{target
model} verifies in parallel. While this mechanism is often studied in centralized
settings, an increasing number of practical deployments place the draft model
on an edge device and the target model on a remote cloud server, making
communication a first-order factor in inference performance \cite{xu2024edgellm}.

This deployment trend is closely related to the broader mobile and edge
computing literature on collaborative intelligence, model partitioning, and
device--cloud co-inference. Prior work has shown that distributing deep neural
network execution across devices and servers can reduce latency and resource
pressure when communication and computation are jointly considered, as in
Neurosurgeon~\cite{neurosurgeon}, JointDNN~\cite{jointdnn}, and distributed DNN
execution over cloud--edge--end
hierarchies~\cite{ddnn}. Recent large-model inference systems such as
Splitwise~\cite{patel2024splitwise} and DistServe~\cite{zhong2024distserve}
further demonstrate that disaggregation and communication-aware scheduling are
important for efficient LLM inference. However, these studies do not answer a
key question specific to distributed speculative decoding: how many draft
tokens should be generated before each verification round under stochastic
network conditions.

This question is non-trivial because draft length directly determines a
system-level tradeoff. A longer draft amortizes round-trip communication delay
across more tokens, but the probability that all drafted tokens are accepted
decreases rapidly with depth, leading to wasted edge computation and
verification work. A shorter draft preserves higher acceptance, but increases
the number of communication rounds and becomes inefficient when network delay is
large or time-varying. As a result, a fixed draft length that works well in one
operating regime can become highly suboptimal when delay conditions drift. While
recent systems work~\cite{sled2025,flexspec2026,configspec2026,venkatesha2025fastedge}
provides engineering evidence and profiling-based heuristics for
communication-aware speculation, a closed-form structural characterization of
the optimal draft length under stochastic communication, together with provable
regret guarantees for an online learning variant, has not been established.

This paper addresses the following key questions:
\begin{itemize}
    \item \textit{Q1: What is the structure of the optimal speculation length
    under communication constraints?} Does an optimal threshold exist, and how
    does it depend on delay?
    \item \textit{Q2: How should the strategy adapt to time-varying network
    conditions?} Can we characterize the value of observing the network state?
    \item \textit{Q3: How can we learn the optimal strategy online when system
    parameters are unknown?}
\end{itemize}

We study distributed speculative decoding as an \emph{online control} problem
for edge-cloud LLM inference. We formalize draft-length selection as a
ratio-type optimal stopping problem~\cite{ferguson2006optimal} that supplies an
analytical foundation for runtime control: at each draft step, the agent decides
whether to continue generating tokens (diminishing marginal returns) or stop
and transmit (incurring stochastic communication cost). This viewpoint enables
structural analysis using sequential decision theory and lattice optimization,
while the theoretical results in Sections~\ref{sec:theory}--\ref{sec:online}
support concrete system design. The main contributions are as follows.

\begin{itemize}
    \item \textit{Communication-aware control formulation.} We formulate
    draft-length selection in distributed speculative decoding as a ratio-type
    optimal stopping problem that explicitly captures the tradeoff between edge
    computation, cloud verification, and stochastic communication delay
    (\S\ref{sec:model}).

    \item \textit{Structural design insights.} We prove that the optimal draft
    length follows a finite delay-monotone threshold policy
    (Theorem~\ref{thm:monotonicity}), identify a critical phase transition below
    which single-token speculation is always optimal
    (Theorem~\ref{thm:phase_transition}), and show that the optimal draft length
    grows only logarithmically with delay (\S\ref{sec:theory}).

    \item \textit{Online runtime adaptation.} We develop UCB-SpecStop, an
    online control algorithm based on a ratio-of-sums estimator aligned
    to~\eqref{eq:arm_cost}, and establish both a gap-dependent logarithmic
    \emph{expected} regret bound and a gap-free
    $O\!\bigl(L_{\max}\sqrt{K_{\max} T\log(K_{\max} T)}\bigr)$ bound when
    system parameters are unknown, for sufficiently large $\beta$
    (Theorem~\ref{thm:regret}; \S\ref{sec:online}).

    \item \textit{State-aware extension.} For time-varying networks, we extend
    the formulation to Markov-modulated channels and show that, under a bounded
    speculation horizon and monotone conditions on the Dinkelbach stopping
    region (Proposition~\ref{thm:markov}), network-state awareness enables a
    state-dependent threshold policy with measurable value-of-information gains
    under the ratio-of-expectations objective (Theorem~\ref{thm:voi};
    \S\ref{sec:theory}--\ref{sec:online}).

    \item \textit{End-to-end edge-cloud validation.} We implement the full
    pipeline on a real edge-cloud testbed and show that the proposed method
    achieves near-oracle performance, outperforms heuristic baselines, and
    avoids the mismatch cost of a single static draft length under delay drift
    (\S\ref{sec:experiments}). We also outline a calibration procedure that maps
    the idealized model to measured round-trip delay and empirical prefix
    acceptance curves.
\end{itemize}

The remainder of this paper is organized as follows. 
Section~\ref{sec:related} reviews related work. 
Section~\ref{sec:model} presents the system model and optimal stopping 
formulation. Section~\ref{sec:theory} develops the main theoretical 
results. Section~\ref{sec:online} introduces the online learning 
algorithm. 
Section~\ref{sec:experiments} reports end-to-end hardware validation
on an edge and cloud testbed.
Section~\ref{sec:conclusion} concludes the paper.

\section{Related Work}
\label{sec:related}

This section reviews four strands of related work: speculative decoding,
collaborative mobile-edge inference and distributed LLM inference, optimal
stopping theory, and communication-aware speculative decoding. For each area, we
highlight the gap that our work addresses.

\subsection{Speculative Decoding}

Speculative decoding has been extensively studied as a lossless
acceleration technique for LLM inference. The draft-and-verify paradigm
traces back to Stern et al.~\cite{stern2018blockwise}, who proposed
blockwise parallel decoding that predicts multiple future positions and
validates the longest acceptable prefix. Leviathan
et al.~\cite{leviathan2023fast} and Chen
et al.~\cite{chen2023accelerating} formalized speculative decoding with
rejection sampling that provably preserves the target distribution.
Subsequent work improves the draft mechanism along several dimensions:
SpecInfer~\cite{miao2024specinfer} introduces tree-structured
speculation; EAGLE~\cite{li2024eagle} leverages feature-level
uncertainty for confidence-based adaptive drafting;
Medusa~\cite{cai2024medusa} replaces the separate draft model with
multiple decoding heads attached to the target model; and
REST~\cite{he2024rest} uses retrieval to construct draft sequences.
Online Speculative Decoding~\cite{liu2024online} adapts the draft model
itself to improve acceptance rates over time.

Most centralized speculative decoding methods treat communication
overhead as negligible; recent edge-cloud extensions (reviewed
in \S\ref{subsec:related_commaware}) introduce network awareness
primarily through profiling and engineering heuristics. Our work
complements these systems by providing a closed-form structural
analysis of the draft-length tradeoff and regret guarantees for
online learning under a simplified stochastic delay model.

\subsection{Distributed LLM Inference}

Beyond speculative decoding itself, our setting connects to collaborative
intelligence and mobile-edge inference. Prior studies split or offload deep
neural network execution across end devices, edge servers, and the cloud to
reduce latency and energy consumption under resource constraints, including
Neurosurgeon~\cite{neurosurgeon}, JointDNN~\cite{jointdnn}, and distributed DNN
execution over cloud--edge--end
hierarchies~\cite{ddnn}. These lines of work established the importance of
communication-aware model partitioning and runtime adaptation in mobile
systems.

Distributed LLM inference systems have emerged to improve hardware 
utilization by disaggregating inference phases across machines, but none 
provides theoretical guarantees for speculation strategy under 
communication uncertainty. Splitwise~\cite{patel2024splitwise} 
disaggregates the prefill and decode phases across machines to improve 
hardware utilization. 
DistServe~\cite{zhong2024distserve} further optimizes goodput by 
co-locating complementary workloads. These systems account for 
inter-machine communication at the engineering level.

While these systems address communication overhead 
implicitly through scheduling heuristics, they do not employ speculative 
decoding or provide theoretical guarantees for inference strategy 
selection under uncertain communication. Our work bridges speculative 
decoding theory with distributed inference by explicitly modeling 
communication randomness. Rather than deciding where to place model layers
or phases, we study how to dynamically control speculative depth, with both
structural properties and online learning guarantees for that control problem.

\subsection{Optimal Stopping Theory}

Optimal stopping theory provides a mature mathematical framework for 
sequential decisions under uncertainty, but has not been applied to the 
speculative decoding setting. Classical foundations~\cite{chow1971great, 
ferguson2006optimal} include Wald's sequential 
analysis~\cite{wald1947sequential} for hypothesis testing with sequential 
observations. Modern applications include optimal search, secretary 
problems, and financial option pricing. Topkis~\cite{topkis1978} 
established monotone comparative statics for optimization problems with 
supermodular structure.

The specific structure of 
our problem, namely a ratio objective with exponentially decaying denominator 
gains and linearly growing numerator costs, does not reduce to standard 
stopping problems and requires tailored analysis.

\subsection{Communication-Aware and Adaptive Drafting}
\label{subsec:related_commaware}

A growing body of work explicitly addresses speculative decoding
in edge-cloud or communication-constrained settings.
SLED~\cite{sled2025} frames speculative decoding as an
edge and server orchestration problem with dynamic drafting and
timeouts under network uncertainty.
Venkatesha et al.~\cite{venkatesha2025fastedge} build a real
edge-cloud speculative decoding system with early exits and
preemptive drafting, demonstrating cost-effective deployment.
Different from these primarily system-driven designs, we emphasize
adaptive draft-length control with closed-form structural analysis and an
online algorithm with regret guarantees under a simplified stochastic delay
model.
FlexSpec~\cite{flexspec2026} introduces channel-aware adaptive
speculation with frozen drafts and evolving targets.
ConfigSpec~\cite{configspec2026} profiles configurations
including speculative length $K^*$, draft model choice,
quantization, and device placement for distributed serving.

On the algorithmic side, SpecDec++~\cite{huang2024specdecpp}
uses a learned acceptance predictor to adapt candidate lengths;
its probability threshold can be viewed as an empirical analogue
of our marginal-cost crossing condition
(Corollary~\ref{cor:closed_form}), specialized to
content-dependent acceptance.
TETRIS~\cite{tetris2025} studies batch speculative decoding and
shows that effective draft depth interacts with batching and
hardware saturation; our framework could incorporate this by
letting verification cost depend on batch size $c_v(b)$, which
we leave for future work.
Batch speculative decoding with correctness
guarantees~\cite{batchspec2025} highlights synchronization
overheads from ragged acceptance in batched settings.

Our work differs from these systems in providing a closed-form
structural theory (threshold optimality, monotonicity, phase
transition, and logarithmic scaling) and regret guarantees for
online learning under a simplified stochastic-delay model. We
view our analytical results and these systems' empirical
evidence as \emph{complementary}: our closed-form $d_c$ and
$k^*(d)$ formulas give a fast configuration prior that systems
such as ConfigSpec~\cite{configspec2026} can refine through
profiling, while platforms such as
SLED~\cite{sled2025} provide
the deployment substrate on which our policies can be evaluated.


\section{System Model and Problem Formulation}
\label{sec:model}

This section formalizes the distributed speculative decoding problem.
Fig.~\ref{fig:system_overview} illustrates the system architecture and
the fundamental tradeoff that motivates our formulation. The edge device
hosts a lightweight draft model $M_d$ that generates $k$ candidate tokens,
which are then transmitted over an unreliable network link to a cloud
server hosting the target model $M_t$ for parallel verification. The
central challenge is selecting the draft length $k$ (\textit{computation-communication tradeoff}): a short draft
maintains high acceptance probability ($\alpha^k \approx 1$) but requires
frequent communication rounds, while a long draft amortizes the stochastic
communication cost $2D$ but suffers exponentially decaying acceptance
($\alpha^k \to 0$). As shown in the bottom of Fig.~\ref{fig:system_overview},
this tradeoff produces a U-shaped cost-per-token curve $C(k,d)$ whose
minimum defines the optimal draft length $k^*$.

\begin{figure}[t]
\centering
\includegraphics[width=0.8\columnwidth]{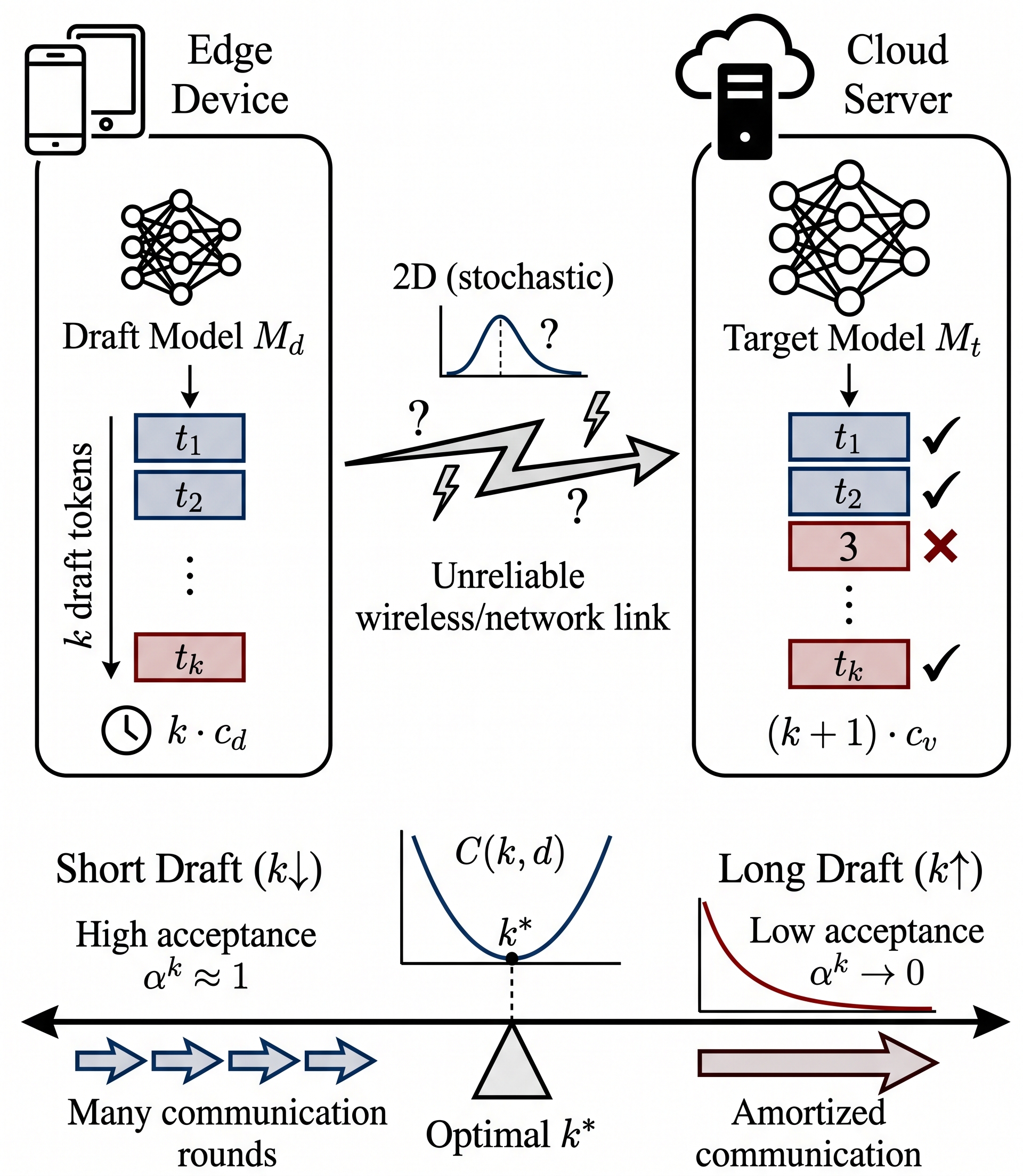}
\caption{Distributed speculative decoding architecture. 
}
\label{fig:system_overview}
\end{figure}

We formalize this architecture (\S\ref{subsec:arch}), introduce the token
acceptance model (\S\ref{subsec:acceptance}), define the cost-per-token objective
(\S\ref{subsec:objective}), and cast the problem as an optimal stopping problem
(\S\ref{subsec:osp}). The key result of this section is Definition~\ref{def:osp},
which establishes the formal framework for all subsequent theoretical analysis.

\subsection{Distributed Inference Architecture}
\label{subsec:arch}

We consider a two-tier distributed inference system illustrated 
conceptually as follows.

\textit{(1) Edge device.} Hosts the draft model $M_d$ (e.g., a 7B-parameter 
model). Each draft token requires $c_d > 0$ time units to generate.

\textit{(2) Cloud server.} Hosts the target model $M_t$ (e.g., a 
70B-parameter model). Verifying $k$ draft tokens and generating one 
bonus token requires $(k+1) \cdot c_v$ time units, where $c_v > 0$ is 
the per-token verification cost. Verification proceeds in parallel 
across all $k$ tokens~\cite{leviathan2023fast}.

\textit{(3) Communication channel.} The edge and cloud communicate over a 
channel with random one-way delay $D \geq 0$. We denote the round-trip 
cost as $2D$. The delay distribution may depend on an observable network 
state $s \in \mathcal{S}$.

\subsection{Token Acceptance Model}
\label{subsec:acceptance}

We model the draft-target acceptance process following the analysis 
in~\cite{leviathan2023fast}. Let $\alpha \in (0,1)$ denote the probability 
that a single draft token is accepted by the target model's rejection 
sampling scheme.

\begin{assumption}[Geometric Acceptance]
\label{asm:geometric}
Token acceptance events are conditionally independent across positions. 
The probability that all $k$ consecutive draft tokens are accepted is 
$\alpha^k$.
\end{assumption}

Assumption~\ref{asm:geometric} is standard in the speculative decoding 
literature~\cite{leviathan2023fast, chen2023accelerating} and captures 
the essential exponential decay structure. While real acceptance rates 
exhibit positional dependence~\cite{li2024eagle}, this model enables 
tractable closed-form analysis. 

Under Assumption~\ref{asm:geometric}, the expected number of accepted 
tokens (including the bonus token generated by the target model upon 
the first rejection) from a draft of length $k$ is:
\begin{equation}
\label{eq:expected_accepted}
    B(k) \triangleq \mathbb{E}[A(k)] = \frac{1 - \alpha^{k+1}}{1 - \alpha}
\end{equation}

We introduce the notation $B(k)$ for conciseness throughout the paper.

\subsection{Cost-Per-Token Objective}
\label{subsec:objective}

For a speculation round with draft length $k$, the 
total elapsed time is:
\begin{equation}
\label{eq:total_time}
    T(k, D) = \underbrace{k \cdot c_d}_{\text{draft generation}} + 
    \underbrace{2D}_{\text{round-trip comm.}} + 
    \underbrace{(k+1) \cdot c_v}_{\text{parallel verification}}
\end{equation}

We measure performance by the expected 
time per accepted token:
\begin{equation}
\label{eq:cost_per_token}
    C(k, d) = \frac{N(k,d)}{B(k)} = 
    \frac{k(c_d + c_v) + 2d + c_v}{(1 - \alpha^{k+1})/(1 - \alpha)}
\end{equation}
where $N(k,d) \triangleq k(c_d + c_v) + 2d + c_v$ denotes the total 
cycle cost and $d = \mathbb{E}[D]$ in the deterministic approximation. 
The goal is to find the draft length $k$ minimizing $C(k,d)$.

\subsection{Optimal Stopping Formulation}
\label{subsec:osp}

We now cast draft length selection as a sequential decision problem.

\begin{definition}[Optimal Stopping for Speculative Decoding]
\label{def:osp}
At each discrete step $n = 1, 2, 3, \ldots$, the agent has generated 
$n$ draft tokens. The agent observes the network state $s_n \in 
\mathcal{S}$ and selects an action:
\begin{itemize}
    \item \textit{Continue} ($a_n = 0$): Generate token $n{+}1$, 
    incurring cost $c_d$.
    \item \textit{Stop} ($a_n = 1$): Transmit $n$ tokens for 
    verification, incurring communication cost $2D(s_n)$ and 
    verification cost $(n{+}1) c_v$.
\end{itemize}
The objective is to find a stopping rule $\tau$ minimizing:
\begin{equation}
\label{eq:objective}
    \min_{\tau} \;\; \frac{\mathbb{E}\bigl[\tau \cdot c_d + 2D(s_\tau) + 
    (\tau+1) \cdot c_v\bigr]}
    {\mathbb{E}\bigl[B(\tau)\bigr]}
\end{equation}
\end{definition}

This is a \textit{ratio-type} optimal stopping 
problem~\cite{ferguson2006optimal}. The standard approach converts it 
to an additive problem via Lagrange multipliers: for a given target 
rate $\lambda$, find $\tau$ minimizing 
$\mathbb{E}[N(\tau, D) - \lambda \cdot B(\tau)]$, then search for 
$\lambda^*$ such that the constraint binds.

In summary, we have established a precise mathematical formulation 
(Definition~\ref{def:osp}) that captures the 
computation and communication tradeoff in distributed speculative decoding. 
The next section derives structural properties of the optimal solution.

\section{Theoretical Results}
\label{sec:theory}

This section presents our main theoretical contributions. We proceed in 
three stages of increasing generality: deterministic delay 
(\S\ref{subsec:deterministic}), stochastic delay 
(\S\ref{subsec:stochastic_delay}), and Markov-modulated channels 
(\S\ref{subsec:markov}). We then present the paper's most surprising 
result, a phase transition and logarithmic scaling law 
(\S\ref{subsec:phase}), followed by a characterization of the value of 
network state information (\S\ref{subsec:voi}). The key takeaways are:
\begin{itemize}
    \item The optimal draft length $k^*$ always exists and is finite 
    (Theorem~\ref{thm:existence}).
    \item $k^*$ is monotonically non-decreasing in delay 
    (Theorem~\ref{thm:monotonicity}).
    \item $k^*$ grows only logarithmically with delay 
    (Theorem~\ref{thm:phase_transition}).
\end{itemize}

\subsection{Deterministic Communication Delay}
\label{subsec:deterministic}

We begin with the simplest setting: the communication delay is a known 
constant $d \geq 0$. This baseline analysis reveals the fundamental 
structure that persists in more general settings.

\begin{theorem}[Existence and Finiteness]
\label{thm:existence}
Under Assumption~\ref{asm:geometric} with $\alpha \in (0,1)$, 
$c_d > 0$, $c_v \geq 0$, and $d \geq 0$, the cost function $C(k,d)$ 
achieves its minimum at a finite $k^* \in \mathbb{N}^+$.
\end{theorem}

\begin{proof}
Since $\alpha\in(0,1)$, we have
\begin{equation}
    B(k)=\frac{1-\alpha^{k+1}}{1-\alpha}\le \frac{1}{1-\alpha}.
\end{equation}
Meanwhile,
\begin{equation}
    N(k,d)=k(c_d+c_v)+2d+c_v
\end{equation}
grows linearly in $k$. Hence $C(k,d)=N(k,d)/B(k)\to\infty$ as $k\to\infty$.
Therefore, there exists $K_0$ such that $C(k,d)>C(1,d)$ for all $k>K_0$.
The minimization over $\mathbb{N}^+$ can thus be restricted to the finite set
$\{1,\ldots,K_0\}$, where the minimum is attained.
\end{proof}

The finiteness of $k^*$ arises from the tension between two forces: 
the communication cost $2d$ favors large $k$ (amortization), while the 
exponential acceptance decay $\alpha^k$ penalizes large $k$ 
(diminishing returns). Neither force dominates asymptotically, and their 
balance determines $k^*$.

We now characterize how $k^*$ responds to changes in communication delay. 
The following result answers \textbf{Q1} from the Introduction.

\begin{theorem}[Delay Monotonicity]
\label{thm:monotonicity}
Let
\begin{equation}
    k^-(d)=\min\arg\min_{k\ge 1} C(k,d)
\end{equation}
be the smallest optimal draft length. Then $k^-(d)$ is non-decreasing in $d$.
Similarly, the largest selector
\begin{equation}
    k^+(d)=\max\arg\min_{k\ge 1} C(k,d)
\end{equation}
is non-decreasing in $d$.
\end{theorem}

\begin{proof}
We apply Topkis' monotone comparative statics 
theorem~\cite{topkis1978}. Define $f(k,d) = -C(k,d)$. We verify the 
\textit{increasing differences} condition: for $k_2 > k_1$,
\begin{equation}
    \frac{\partial}{\partial d}\bigl[C(k_1, d) - C(k_2, d)\bigr] = 
    \frac{2\bigl(B(k_2) - B(k_1)\bigr)}{B(k_1)\,B(k_2)} > 0.
\end{equation}
The positivity follows because $B(\cdot)$ is strictly increasing. This 
establishes that the \textit{relative} advantage of larger $k$ grows 
with $d$.
Although the action space is $\mathbb{N}^+$, which is not compact,
Theorem~\ref{thm:existence} implies that for any compact delay interval
$d\in[d_1,d_2]$, every minimizer of $C(\cdot,d)$ lies in a finite subset of
$\mathbb{N}^+$. Topkis' theorem therefore applies on this finite lattice, and the
resulting monotonicity of the smallest and largest minimizers extends to all
$d\ge 0$.
\end{proof}

When communication is expensive (large $d$), the fixed cost $2d$ 
dominates the per-cycle budget. The agent should ``batch'' more draft 
tokens per round to amortize this fixed cost, accepting the reduced 
marginal acceptance probability. This is analogous to the economic 
principle that higher fixed costs favor larger production batches.

The following lemma and corollary provide an operational stopping rule.

\begin{lemma}[Discrete Unimodality]
\label{lem:quasi_convex}
Under Assumption~\ref{asm:geometric}, for fixed $d$, the objective $C(k,d)$ is
discrete quasi-convex in $k$. Therefore, the first integer $k$ satisfying
$C(k{+}1,d)\ge C(k,d)$ is a global minimizer.
\end{lemma}

\begin{proof}
Let
\begin{equation}
    a \triangleq c_d+c_v,\qquad b \triangleq 2d+c_v,
\end{equation}
so $N(k,d)=ak+b$ and $B(k)=\frac{1-\alpha^{k+1}}{1-\alpha}$.
The condition $C(k{+}1,d)\ge C(k,d)$ is equivalent to
\begin{equation}
    a\,B(k)\ge N(k,d)\alpha^{k+1}.
\end{equation}
Define
\begin{equation}
    H(k)\triangleq \frac{aB(k)}{\alpha^{k+1}}-N(k,d)
    =\frac{a}{1-\alpha}\!\left(\alpha^{-(k+1)}-1\right)-ak-b.
\end{equation}
Then
\begin{equation}
    H(k{+}1)-H(k)=a\!\left(\alpha^{-(k+2)}-1\right)>0,
\end{equation}
so $H(k)$ is strictly increasing. Hence the crossing condition can occur at
most once, i.e., $C(k,d)$ decreases and then increases. Therefore, the first
$k$ satisfying $C(k{+}1,d)\ge C(k,d)$ is globally optimal.
\end{proof}

\begin{corollary}[Marginal Cost Stopping Rule]
\label{cor:closed_form}
Suppose that $C(k,d)$ is discrete quasi-convex in $k$. Then the smallest global
minimizer is the first integer satisfying:
\begin{equation}
\label{eq:simple_condition}
    C(k, d) \leq \frac{c_d + c_v}{\alpha^{k+1}}
\end{equation}
Equivalently: \textit{stop when the current average cost per token 
falls below the marginal cost of generating one additional draft token}.
\end{corollary}

\begin{proof}
The condition $C(k{+}1,d) \geq C(k,d)$ simplifies to 
$(c_d + c_v) \cdot B(k) \geq N(k,d) \cdot \alpha^{k+1}$, 
which rearranges to~\eqref{eq:simple_condition}. The derivation uses 
$N(k{+}1,d) = N(k,d) + (c_d{+}c_v)$ and 
$B(k{+}1) = B(k) + \alpha^{k+1}$.
By Lemma~\ref{lem:quasi_convex}, the first such crossing is globally optimal.
\end{proof}

Condition~\eqref{eq:simple_condition} is the discrete analog of the 
classical ``marginal cost equals average cost'' optimality condition. 
The marginal cost $\frac{c_d + c_v}{\alpha^{k+1}}$ grows exponentially 
(reflecting worsening acceptance odds), while the average cost $C(k,d)$ 
first decreases (amortizing the communication fixed cost) then 
increases. Their crossing defines $k^*$. This answers \textbf{Q1}: the 
optimal strategy has an intuitive threshold structure with a precise 
mathematical characterization.

\subsection{Stochastic Communication Delay}
\label{subsec:stochastic_delay}

We now relax the deterministic delay assumption. Let $D$ be a 
non-negative random variable with mean $\mu_D$ and the stopping 
decision is committed \textit{before} observing the delay realization 
(the ``commit-before-observing'' model).

\begin{theorem}[Mean-Sufficiency under Commit-Before-Observing]
\label{thm:stochastic}
Under the commit-before-observing model, the optimal draft length 
depends on the delay distribution only through its mean:
\begin{equation}
    k^* = \arg\min_{k \geq 1} \frac{k(c_d + c_v) + 2\mu_D + c_v}{B(k)}
\end{equation}
\end{theorem}

\begin{proof}
Under commit-before-observing, $B(k)$ is deterministic for fixed $k$, so
\begin{equation}
    \frac{\mathbb{E}[N(k,D)]}{\mathbb{E}[B(k)]}
    =\frac{k(c_d+c_v)+2\mu_D+c_v}{B(k)}
    = C(k,\mu_D).
\end{equation}
Thus the ratio-of-expectations objective depends on the delay law only through
$\mu_D$, reducing optimization to the deterministic-mean case.
\end{proof}

Theorem~\ref{thm:stochastic} shows that under commitment, only the
mean delay affects the optimal strategy for the ratio-of-expectations
objective, while higher moments are irrelevant in this information structure.
It does not claim higher-order moments are irrelevant for state-observable or
adaptive stopping models.
This \textit{changes fundamentally} when the agent can observe network 
state before deciding, as we show next. The gap between these two 
information structures quantifies the value of real-time network 
monitoring.

\subsection{Markov-Modulated Extension}
\label{subsec:markov}

We now model time-varying network conditions. The network state 
$\{s_t\}$ evolves as a Markov chain, and the agent observes $s_n$ 
after generating draft token $n$, using this information to decide 
whether to stop. This addresses \textbf{Q2}. Throughout, we align the
sequential problem with the bandit section by imposing a maximum draft
depth $K_{\max}$ (i.e., mandatory stop at $\tau=K_{\max}$), so that stopping
times satisfy $\tau\le K_{\max}$ and match the arm set
$\mathcal{K}=\{1,\ldots,K_{\max}\}$ in Section~\ref{sec:online}.

\begin{assumption}[Markov-Modulated Channel]
\label{asm:markov}
The network state $\{s_t\}_{t \geq 1}$ is a finite-state Markov chain 
on $\mathcal{S} = \{1, \ldots, S\}$ with transition matrix $P$.
We assume:
\begin{enumerate}
    \item[(a)] \emph{Monotone mean delay}:
    $d(s) = \mathbb{E}[D \mid s_t = s]$ is non-decreasing in $s$
    (states ordered from low to high delay).
    \item[(b)] \emph{Stochastic monotonicity of $P$}: for every
    non-decreasing function $h: \mathcal{S} \to \mathbb{R}$, the
    map $s \mapsto \sum_{s'} P(s' \mid s)\, h(s')$ is
    non-decreasing in $s$. Equivalently, $P(\,\cdot\, \mid s)$
    is stochastically increasing in $s$ in the usual stochastic
    order.
\end{enumerate}
\end{assumption}

\begin{remark}
Condition (b) is standard for monotone MDPs~\cite{topkis1978}
and holds, for example, for birth and death chains modeling
congestion levels and for any tridiagonal $P$ whose row
distributions are stochastically ordered. Without (b), worse
states can transition to better states faster than better states
do, which can break the monotonicity of value functions.
\end{remark}

\begin{proposition}[State-Dependent Threshold under a Bounded Horizon]
\label{thm:markov}
Assume Assumptions~\ref{asm:geometric} and~\ref{asm:markov}. Fix $K_{\max}<\infty$
and require the agent to stop no later than $K_{\max}$ (i.e., $\tau\le K_{\max}$).
For $\lambda\ge 0$, define the $\lambda$-penalized cost after $n$ draft tokens in state $s$ by
\begin{equation}
\label{eq:g_lambda}
    g_\lambda(n,s)\triangleq n c_d + 2d(s) + (n{+}1)c_v - \lambda\, B(n),
\end{equation}
the finite-horizon value functions
\begin{align}
    V_\lambda(K_{\max},s) &= g_\lambda(K_{\max},s), \notag\\
\label{eq:V_lambda}
    V_\lambda(n,s) &=
    \min\!\left\{
    g_\lambda(n,s),\,
    \sum_{s'} P(s'\mid s)\, V_\lambda(n{+}1,s')
    \right\},
    \qquad 1\le n<K_{\max},
\end{align}
and the $Q$-factors
\begin{equation}
\label{eq:Q_factors_markov}
\begin{aligned}
    Q_\lambda^{\mathrm{stop}}(n,s) &\triangleq g_\lambda(n,s),\\
    Q_\lambda^{\mathrm{cont}}(n,s) &\triangleq
    \sum_{s'} P(s'\mid s)\, V_\lambda(n{+}1,s').
\end{aligned}
\end{equation}
For $1\le n<K_{\max}$, define the stopping advantage
\begin{equation}
    \Gamma_\lambda(n,s)\triangleq
    Q_\lambda^{\mathrm{cont}}(n,s)-Q_\lambda^{\mathrm{stop}}(n,s),
\end{equation}
and set $\Gamma_\lambda(K_{\max},s)=+\infty$ for all $s$, encoding mandatory stop at $K_{\max}$.
Suppose that \emph{(i)} for each $s$, $\Gamma_\lambda(n,s)$ is nondecreasing in
$n\in\{1,\ldots,K_{\max}\}$; and \emph{(ii)} the stopping region
$\mathcal{R}_\lambda \triangleq \{(n,s):\Gamma_\lambda(n,s)\ge 0\}$
is decreasing in $s$: if $(n,s)\in\mathcal{R}_\lambda$ and $s'\le s$, then $(n,s')\in\mathcal{R}_\lambda$.
Then an optimal rule for the $\lambda$-penalized problem is a state-dependent threshold policy:
with
\begin{equation}
    k_\lambda^*(s)\triangleq
    \min\bigl\{n\in\{1,\ldots,K_{\max}\}:\Gamma_\lambda(n,s)\ge 0\bigr\},
\end{equation}
it is optimal to continue while the current draft length satisfies $n<k_\lambda^*(s)$ and to stop once $n\ge k_\lambda^*(s)$.
Moreover,
\begin{equation}
    s'\le s \quad\Longrightarrow\quad k_\lambda^*(s')\le k_\lambda^*(s).
\end{equation}
Let $\lambda^*$ be the Dinkelbach parameter for~\eqref{eq:objective} restricted to $\tau\le K_{\max}$.
Then $k_{\lambda^*}^*(s)$ is optimal for the original ratio objective among stopping rules with $\tau\le K_{\max}$.
\end{proposition}

\begin{proof}
Fix $\lambda$. The finite-horizon dynamic program~\eqref{eq:V_lambda} is well defined because $\mathcal{S}$ and the action horizon are finite.
For $1\le n<K_{\max}$, stopping is optimal iff
$Q_\lambda^{\mathrm{stop}}(n,s)\le Q_\lambda^{\mathrm{cont}}(n,s)$, equivalently $\Gamma_\lambda(n,s)\ge 0$.
By hypothesis~(i), for each $s$ the set $\{n:\Gamma_\lambda(n,s)\ge 0\}$ is an upper interval in $\{1,\ldots,K_{\max}\}$, and it is nonempty because $\Gamma_\lambda(K_{\max},s)=+\infty$.
Hence $k_\lambda^*(s)$ is well defined.

If $s'\le s$, then $(k_\lambda^*(s),s)\in\mathcal{R}_\lambda$, so hypothesis~(ii) implies $(k_\lambda^*(s),s')\in\mathcal{R}_\lambda$ and therefore $k_\lambda^*(s')\le k_\lambda^*(s)$.
If $k_\lambda^*(s)=K_{\max}$, the same inequality is trivial.

Finally, restrict~\eqref{eq:objective} to $\tau\le K_{\max}$. The policy class is finite and $\mathbb{E}[B(\tau)]\ge 1$.
Define the Dinkelbach objective
\begin{equation}
    J(\lambda)\triangleq \min_{\tau\le K_{\max}}
    \mathbb{E}\bigl[\tau c_d + 2d(s_\tau)+(\tau+1)c_v-\lambda B(\tau)\bigr].
\end{equation}
Then $J$ is continuous and strictly decreasing on $[0,\infty)$, the unique root $\lambda^*$ satisfies $J(\lambda^*)=0$, and a policy minimizing the $\lambda^*$-penalized DP also minimizes the ratio
$\mathbb{E}[\tau c_d + 2d(s_\tau)+(\tau+1)c_v]/\mathbb{E}[B(\tau)]$ over $\tau\le K_{\max}$~\cite{dinkelbach1967}.
\end{proof}

Equation~\eqref{eq:V_lambda} uses a pure total-cost recursion: stopping pays the
current $\lambda$-penalized total cost, while continuing transitions to the next
state and inherits its total-cost value. This avoids mixing total and
incremental accounting in the Bellman recursion.

The single-crossing structure in $n$ is natural in this setting: larger $n$ reduces
the marginal acceptance gain through $\alpha^{n+1}$, while larger $s$ increases
the communication cost saved by continuing. Hypothesis~(ii) on $\mathcal{R}_\lambda$
rules out pathological level shifts in $\Gamma_\lambda(\cdot,s)$ across $s$ that
would reorder stopping thresholds; a sufficient analytic condition is that
$\Gamma_\lambda(n,s)$ be non-increasing in $s$ for each $n$.
This condition is consistent with
the U-shaped empirical cost curves observed in Section~\ref{sec:experiments}.

\subsection{Phase Transition and Logarithmic Scaling}
\label{subsec:phase}

The preceding results establish that $k^*$ increases with $d$. A natural follow-up question is: \textit{how fast?} The following theorem is our most surprising result and shows the growth is remarkably slow.

\begin{theorem}[Critical Delay and Logarithmic Scaling]
\label{thm:phase_transition}
Assume the discrete quasi-convexity condition in
Lemma~\ref{lem:quasi_convex} and define the critical delay:
\begin{equation}
\label{eq:dc}
    d_c = \frac{(c_d{+}c_v)(1{+}\alpha)}{2\alpha^2} - 
    \frac{c_d + 2c_v}{2}
\end{equation}
If $d_c>0$, then:
\begin{enumerate}
    \item \textit{Sub-critical regime} ($d < d_c$): letting
    $k^-{:=}\min\arg\min_{k\ge 1} C(k,d)$, one has $k^-{=}1$. Moreover, when the
    first crossing is strict (equivalently $C(1,d)<C(2,d)$), the minimizer is unique.
    \item \textit{Boundary} ($d = d_c$): $k=1$ and $k=2$ are both
    optimal under the local crossing condition.
    \item \textit{Asymptotic regime} (large $d$): letting
    $k^-({d}){:=}\min\arg\min_{k\ge 1} C(k,d)$,
    \begin{equation}
        k^-(d)=\Theta\!\left(\frac{\log d}{\log(1/\alpha)}\right).
    \end{equation}
    The same $\Theta(\cdot)$ scaling holds for every optimal selector
    $k\in\arg\min_{k\ge 1} C(k,d)$: since $H(\cdot;d)$ in
    Lemma~\ref{lem:quasi_convex} is strictly increasing, the sign change of
    $C(k{+}1,d)-C(k,d)$ occurs at most once, so the argmin is either a single
    integer or two adjacent integers; hence every optimal selector shares the
    same $\Theta\bigl(\log d/\log(1/\alpha)\bigr)$ scaling.
\end{enumerate}
If $d_c\le 0$, the system is already in the post-transition regime at zero delay.
\end{theorem}

\begin{proof}
Set $a\triangleq c_d+c_v>0$ and $b\triangleq 2d+c_v$, so $N(k,d)=ak+b$ and
$B(k)=(1-\alpha^{k+1})/(1-\alpha)$.
By Lemma~\ref{lem:quasi_convex}, $k^-(d)$ is the first integer $k\ge 1$ such that
$C(k{+}1,d)\ge C(k,d)$, equivalently
\begin{equation}
\label{eq:phase_cross}
    a\,B(k)\ge (ak+b)\alpha^{k+1}.
\end{equation}
Define
\begin{equation}
    H(k;d)\triangleq \frac{a\,B(k)}{\alpha^{k+1}}-(ak+b)
    =\frac{a}{1-\alpha}\Bigl(\alpha^{-(k+1)}-1\Bigr)-ak-b.
\end{equation}
As in the proof of Lemma~\ref{lem:quasi_convex},
\begin{equation}
    H(k{+}1;d)-H(k;d)=a\bigl(\alpha^{-(k+2)}-1\bigr)>0,
\end{equation}
so $H(\cdot;d)$ is strictly increasing and $k^-(d)=\min\{k\ge 1:H(k;d)\ge 0\}$.

The condition $k^-(d)=1$ is equivalent to $H(1;d)\ge 0$, i.e., to
$C(1,d)\le (c_d+c_v)/\alpha^2$.
Substituting $C(1,d)=(c_d+2d+2c_v)/(1+\alpha)$ and solving for $d$ yields $d\le d_c$
with $d_c$ in~\eqref{eq:dc}. Hence, if $d<d_c$, then $k^-(d)=1$, and when the first
crossing is strict ($C(1,d)<C(2,d)$), the minimizer is unique.
If $d=d_c$, then $C(1,d)=C(2,d)$ while strict increase of $H(\cdot;d)$ implies
$C(k,d)>C(1,d)$ for all $k\ge 3$, so $k=1$ and $k=2$ are both optimal.

For logarithmic scaling, let $r\triangleq 1/\alpha>1$.
Since $H(k^-(d);d)\ge 0$,
\begin{equation}
    \frac{a}{1-\alpha}\Bigl(r^{k^-(d)+1}-1\Bigr)\ge a k^-(d)+b\ge b=2d+c_v,
\end{equation}
so
\begin{equation}
\begin{aligned}
    r^{k^-(d)+1}&\ge 1+\frac{(1-\alpha)(2d+c_v)}{a},\\
    k^-(d)&\ge \frac{\log\bigl(1+(1-\alpha)(2d+c_v)/a\bigr)}{\log r}-1
    =\Omega\!\left(\frac{\log d}{\log(1/\alpha)}\right).
\end{aligned}
\end{equation}

For an upper bound, fix any
\begin{equation}
    M>\frac{2(1-\alpha)}{ar}
\end{equation}
and define
\begin{equation}
    K_d\triangleq \left\lceil
    \frac{\log\bigl(M(2d+c_v)\bigr)}{\log r}
    \right\rceil.
\end{equation}
Then $r^{K_d+1}\ge rM(2d+c_v)$, hence
\begin{equation}
\begin{aligned}
    H(K_d;d)
    &\ge \frac{a}{1-\alpha}\Bigl(rM(2d+c_v)-1\Bigr)-aK_d-(2d+c_v)\\
    &=\bigl(2d+c_v\bigr)\left(\frac{arM}{1-\alpha}-1\right)-aK_d
    -\frac{a}{1-\alpha}.
\end{aligned}
\end{equation}
By the choice of $M$, the parenthesis is strictly positive, while $K_d=O(\log d)$.
Thus $H(K_d;d)\ge 0$ for all sufficiently large $d$, so
$k^-(d)\le K_d=O\bigl(\log d/\log(1/\alpha)\bigr)$.
Combining the bounds yields the $\Theta(\cdot)$ claim for $k^-(d)$, and the same
order holds for any minimizer because, as above, strict monotonicity of
$H(\cdot;d)$ implies the argmin is a singleton or two adjacent integers.
\end{proof}

Theorem~\ref{thm:phase_transition} is the central insight of this paper. 
It reveals three non-obvious facts:

(1) Doubling communication 
delay adds only $ 1/\!\log(1/\alpha)$ tokens to the optimal 
draft length. For typical $\alpha = 0.7$, this is merely $2.8$ 
extra tokens. The intuition is that exponential acceptance decay provides 
a natural ``damping'' that prevents the optimal response from 
overreacting to delay increases.

(2) Below $d_c$, communication is cheap 
enough that single-token speculation (maximum verification frequency) 
is optimal regardless of further delay reduction. This provides a 
concrete threshold for system designers: if $d < d_c$, there is no 
benefit from multi-token speculation.

(3) The logarithmic scaling means that 
even in extremely high-delay scenarios (e.g., satellite links), the 
optimal draft length remains moderate (e.g., $k^* \approx 10$ to $15$ 
for $d = 500$ms, $\alpha = 0.7$). System implementations need not 
support arbitrarily long speculation buffers.

\subsection{Value of Network State Information}
\label{subsec:voi}

Proposition~\ref{thm:markov} characterizes the structure of state-aware
policies. We now quantify when observing network state can improve performance,
answering the second part of \textbf{Q2}.

\begin{theorem}[Value of State Information under the Ratio Objective]
\label{thm:voi}
Consider a contextual model in which the network state $s\in\mathcal{S}$ is observed
at the beginning of each speculation round and remains fixed during that round.
Let $\{\pi_s\}_{s\in\mathcal{S}}$ be the state distribution.
For a state-dependent policy $\kappa:\mathcal{S}\to\{1,\ldots,K_{\max}\}$, define the
ratio-of-expectations cost
\begin{equation}
\label{eq:C_ctx}
    C_{\mathrm{ctx}}(\kappa)\triangleq
    \frac{\sum_{s\in\mathcal{S}}\pi_s\, N\!\bigl(\kappa(s),d(s)\bigr)}
    {\sum_{s\in\mathcal{S}}\pi_s\, B\!\bigl(\kappa(s)\bigr)}.
\end{equation}
Let
\begin{equation}
    C_{\mathrm{ctx}}^\star\triangleq
    \min_{\kappa:\mathcal{S}\to\{1,\ldots,K_{\max}\}} C_{\mathrm{ctx}}(\kappa).
\end{equation}
For a blind fixed policy $k\in\{1,\ldots,K_{\max}\}$, define
\begin{equation}
\label{eq:C_blind_star}
    C_{\mathrm{blind}}(k)\triangleq
    \frac{\sum_{s\in\mathcal{S}}\pi_s\, N(k,d(s))}{B(k)}
    = C(k,\mu_D),
    \mu_D\triangleq \sum_{s\in\mathcal{S}}\pi_s\, d(s),
\end{equation}
and $C_{\mathrm{blind}}^\star\triangleq \min_{k} C_{\mathrm{blind}}(k)$.
Then
\begin{equation}
\label{eq:voi_def}
    \mathrm{VOI}\triangleq C_{\mathrm{blind}}^\star-C_{\mathrm{ctx}}^\star \ge 0.
\end{equation}
Moreover, the inequality is strict if and only if there exists $\kappa$ such that
$C_{\mathrm{ctx}}(\kappa)<C_{\mathrm{blind}}^\star$.
\end{theorem}

\begin{proof}
Every blind policy $k$ corresponds to the constant mapping $\kappa_k(s)\equiv k$, so
\begin{equation}
    C_{\mathrm{ctx}}^\star
    \le \min_k C_{\mathrm{ctx}}(\kappa_k).
\end{equation}
For $\kappa_k$,
\begin{equation}
\begin{aligned}
    C_{\mathrm{ctx}}(\kappa_k)
    &=\frac{\sum_s \pi_s\, N(k,d(s))}{\sum_s \pi_s\, B(k)}
    =\frac{\sum_s \pi_s\, N(k,d(s))}{B(k)}\\
    &=C_{\mathrm{blind}}(k),
\end{aligned}
\end{equation}
where we used $\sum_s\pi_s=1$ and $B(k)$ not depending on $s$.
Hence $C_{\mathrm{ctx}}^\star\le C_{\mathrm{blind}}^\star$, i.e., $\mathrm{VOI}\ge 0$.
Strict positivity holds exactly when some (possibly nonconstant) $\kappa$ achieves a
strictly smaller ratio than every constant policy.
\end{proof}

Our structural results rely on the acceptance model to varying degrees. 
Theorem~\ref{thm:existence} (existence/finiteness) requires only that 
$B(k)$ remain bounded as $k \to \infty$, i.e., that acceptance 
probabilities decay fast enough. Theorem~\ref{thm:monotonicity} 
requires $B(k)$ strictly increasing in $k$, hence positive per-position 
acceptance. Theorem~\ref{thm:phase_transition} and the logarithmic 
envelope use exponential decay $\alpha^k$; for heavier-tailed 
acceptance (e.g., $q_k \propto (1{+}k)^{-\beta}$), scaling becomes 
super-logarithmic. Corollary~\ref{cor:closed_form} holds whenever 
$C(k,d)$ is eventually increasing in $k$ (quasi-convexity), which is 
typical when marginal acceptance decays faster than marginal draft 
cost grows. Extensions to content-dependent, non-stationary acceptance 
are important future work.

This theorem provides a decision criterion for whether to invest in 
network monitoring infrastructure: if the network delay variance is 
small (e.g., wired datacenter interconnect), a fixed $k$ suffices; if 
the variance is large (e.g., cellular networks), state-adaptive 
strategies provide measurable benefit.

In summary, this section establishes that the optimal speculation 
strategy under communication constraints is a threshold policy with 
logarithmic delay sensitivity and a sharp phase transition. The next 
section addresses the practical challenge of unknown system parameters.

\section{Online Learning Algorithm}
\label{sec:online}

This section addresses \textbf{Q3}: learning the optimal draft length
when both the acceptance rate $\alpha$ and delay distribution are
unknown, a practical scenario arising when deploying on new hardware
or connecting to unfamiliar network environments. The key technical
challenge is that our objective~\eqref{eq:arm_cost} is a \textit{ratio
of expectations} rather than a simple expected reward, which invalidates
standard bandit algorithms designed for additive objectives. Naively
minimizing the per-round ratio $N_t/A_t$ optimizes a different criterion
(due to Jensen's inequality), necessitating a ratio-of-sums estimator
with carefully designed confidence bounds.

We formulate the problem as a multi-armed bandit
(\S\ref{subsec:bandit_setup}), present the UCB-SpecStop algorithm
(\S\ref{subsec:ucb_alg}), prove its regret guarantee
(\S\ref{subsec:regret_proof}), and extend it to the contextual setting
(\S\ref{subsec:contextual}). The main result is Theorem~\ref{thm:regret},
which establishes gap-dependent and gap-free \emph{expected} regret bounds of order
$O\!\bigl(L_{\max}\sqrt{K_{\max} T \log(K_{\max} T)}\bigr)$ despite the
non-standard ratio structure (for sufficiently large $\beta$).

\subsection{Bandit Formulation}
\label{subsec:bandit_setup}

The theoretical results of Section~\ref{sec:theory} assume known
$\alpha$, $c_d$, $c_v$, and delay distribution. In practice, these
parameters are difficult to estimate \textit{a priori}: the acceptance
rate $\alpha$ depends on the specific draft and target model pair and
input distribution, while network delay statistics vary across
deployment environments and time. We therefore seek an online algorithm
that learns the optimal draft length through repeated interaction.

At each decoding round $t = 1, \ldots, T$:
\begin{enumerate}
    \item The agent selects $k_t \in \mathcal{K} = \{1, \ldots, K_{\max}\}$.
    \item Nature reveals the realized numerator and denominator
    \begin{align}
        N_t &= k_t(c_d{+}c_v) + 2D_t + c_v, \label{eq:Nt}\\
        A_t &= \text{number of accepted tokens in the round}, \label{eq:At}
    \end{align}
    with $D_t$ the realized delay. The per-round \textit{ratio}
    $N_t/A_t$ is a random variable whose expectation generally differs
    from $\mathbb{E}[N_t]/\mathbb{E}[A_t]$ (Jensen's inequality).
\end{enumerate}

A subtle but crucial point is the choice of optimization criterion.
The design objective is the \textit{ratio of expectations}:
\begin{equation}
\label{eq:arm_cost}
    C(k) \triangleq \frac{\mathbb{E}[N_t \mid k_t = k]}
    {\mathbb{E}[A_t \mid k_t = k]},
\end{equation}
which aligns with~\eqref{eq:cost_per_token} when $(D_t, A_t)$ are
drawn from the same generative model for each fixed $k$. Minimizing
$\mathbb{E}[N_t/A_t \mid k_t{=}k]$ instead would optimize the
\textit{expectation of the ratio}, which is a different criterion that
overweights rounds with few accepted tokens. We therefore use a
\textit{ratio-of-sums} estimator that directly targets~\eqref{eq:arm_cost}.

\begin{definition}[Regret]
\label{def:regret}
Let $k^* \in \arg\min_{k \in \mathcal{K}} C(k)$. The cumulative regret is
\begin{equation}
    \label{eq:regret_def}
    R(T) = \sum_{t=1}^T \bigl( C(k_t) - C(k^*) \bigr).
\end{equation}
\end{definition}

The regret definition charges the \textit{expected} suboptimality gap
$C(k_t) - C(k^*)$ at each round, even though the agent observes only
stochastic realizations $(N_t, A_t)$. This is standard in the bandit
literature~\cite{auer2002finite} and measures the cumulative cost of
exploration.

\subsection{UCB-SpecStop Algorithm}
\label{subsec:ucb_alg}

The core design challenge is constructing a reliable confidence interval
for the ratio estimator $\hat{C}(k) = S_N(k)/S_A(k)$. Unlike standard
UCB1 where the sample mean directly estimates the objective, here the
ratio of two random quantities requires joint concentration control.

For each arm $k$, we maintain cumulative statistics $S_N(k) =
\sum_{t: k_t = k} N_t$ and $S_A(k) = \sum_{t: k_t = k} A_t$, and let
$T_k$ count how many times arm $k$ is played.
The ratio-of-sums estimator $\hat{C}(k) = S_N(k)/S_A(k)$ is consistent
for $C(k)$ by the strong law of large numbers applied to numerator and
denominator separately. Crucially, this estimator avoids the bias
inherent in averaging per-round ratios $\frac{1}{T_k}\sum N_t/A_t$.

Algorithm~\ref{alg:ucb} combines this estimator with a confidence width
that matches the confidence scale in Lemma~\ref{lem:ratio_concentration} with
$\delta\asymp T^{-2}$ (cf.~\eqref{eq:Lmax},~\eqref{eq:ucb_index_scale}).

\begin{algorithm}[t]
\caption{UCB-SpecStop}
\label{alg:ucb}
\begin{algorithmic}[1]
\REQUIRE Arms $\mathcal{K} = \{1, \ldots, K_{\max}\}$, horizon $T$,
confidence parameter $\beta>0$ (Theorem~\ref{thm:regret} assumes $\beta\ge c$
with $c$ from Lemma~\ref{lem:ratio_concentration})
\ENSURE Estimated optimal draft length $\hat{k}^*$
\STATE \textbf{Initialize:} $S_N(k) \leftarrow 0$, 
$S_A(k) \leftarrow 0$, $T_k \leftarrow 0$ for all $k \in \mathcal{K}$
\FOR{$t = 1, 2, \ldots, T$}
    \IF{there exists $k$ with $T_k = 0$}
        \STATE $k_t \leftarrow$ any such $k$
    \ELSE
        \STATE $k_t \leftarrow \arg\min_{k \in \mathcal{K}} 
        \left[\dfrac{S_N(k)}{S_A(k)} -         \beta\,
        L_{\max}\sqrt{\dfrac{\log(4K_{\max}T^2)}{T_k}}\right]$
    \ENDIF
    \STATE Play $k_t$; observe $N_t$ and $A_t$
    \STATE $S_N(k_t) \leftarrow S_N(k_t) + N_t$;\;
    $S_A(k_t) \leftarrow S_A(k_t) + A_t$;\;
    $T_{k_t} \leftarrow T_{k_t} + 1$
\ENDFOR
\RETURN $\hat{k}^* \leftarrow \arg\min_{k \in \mathcal{K}} S_N(k)/S_A(k)$
\end{algorithmic}
\end{algorithm}

The selection rule in Line~6 can be
understood as an optimistic estimate: the agent subtracts a confidence
bonus from the cost estimate, selecting the arm with the lowest
plausible cost. The $\sqrt{\log(4K_{\max}T^2)/T_k}$ term balances exploration
(trying under-sampled arms) against exploitation (playing the
empirically best arm). The division by $S_A(k)$ rather than $T_k$
aligns the estimate with the ratio-of-expectations objective by normalizing
cumulative cost with cumulative accepted tokens. The confidence bonus uses
the $T_k^{-1/2}$ scale from Lemma~\ref{lem:ratio_concentration}.

\subsection{Regret Guarantee}
\label{subsec:regret_proof}

\begin{assumption}[Bounded costs]
\label{asm:bounded}
The delays are bounded: $D \leq D_{\max}$ almost surely. With 
$K_{\max} < \infty$, we have $N_t \leq N_{\max} \triangleq 
K_{\max}(c_d{+}c_v) + 2D_{\max} + c_v$. Moreover, since each round accepts at
least one token (the bonus token) and at most $k_t{+}1 \le K_{\max}{+}1$
tokens, we have $1 \le A_t \le K_{\max}{+}1$ almost surely, so
$\mathbb{E}[A_t \mid k_t=k] \geq B_{\min} \triangleq 1$.
For each fixed arm $k$, observations $(N_t,A_t)$ are independent (or
conditionally independent given $k_t{=}k$).
\end{assumption}

The lower bound $A_t \geq 1$ is guaranteed by the speculative decoding protocol:
even if all $k$ draft tokens are rejected, the target model generates one bonus
token~\cite{leviathan2023fast}. The upper bound $A_t \le K_{\max}+1$ follows
because at most $k_t+1$ tokens can be accepted in one round.
Together these bounds control denominator fluctuations in the ratio estimator.
Define $A_{\max}\triangleq K_{\max}+1$ (matching the almost-sure bound on $A_t$)
and the ratio-concentration scale
\begin{equation}
\label{eq:Lmax}
    L_{\max}\;\triangleq\;\frac{N_{\max}}{B_{\min}}+\frac{N_{\max}A_{\max}}{B_{\min}^2}.
\end{equation}
This is the natural linearized magnitude for perturbations of
$S_N(k)/S_A(k)$ when controlling both numerator and denominator fluctuations.

\begin{lemma}[Uniform Concentration of the Ratio-of-Sums Estimator]
\label{lem:ratio_concentration}
Fix an arm $k$ and let $(N_{k,i},A_{k,i})_{i\ge 1}$ denote the i.i.d.\ observations
generated whenever arm $k$ is pulled.
Assume $0\le N_{k,i}\le N_{\max}$ and $B_{\min}\le A_{k,i}\le A_{\max}$ almost surely
with $B_{\min}>0$. Define $\mu_N(k)=\mathbb{E}[N_{k,1}]$, $\mu_A(k)=\mathbb{E}[A_{k,1}]$,
and $C(k)=\mu_N(k)/\mu_A(k)$.
After $n\ge 1$ pulls of arm $k$, let
\begin{equation}
    \widehat C_{k,n}\triangleq \frac{\sum_{i=1}^n N_{k,i}}{\sum_{i=1}^n A_{k,i}}.
\end{equation}
Under Assumption~\ref{asm:bounded}, there exists a universal constant $c>0$ such that,
with probability at least $1-\delta$, for all $k\in\{1,\ldots,K_{\max}\}$ and all
$1\le n\le T$,
\begin{equation}
\label{eq:ratio_conc}
    \bigl|\widehat C_{k,n}-C(k)\bigr|
    \le
    c\,L_{\max}\sqrt{\frac{\log(4K_{\max}T/\delta)}{n}},
\end{equation}
where $L_{\max}$ is given in~\eqref{eq:Lmax}.
The adaptivity of the arm-selection rule does not invalidate~\eqref{eq:ratio_conc}:
for each $k$, the subsequence observed on pulls of arm $k$ equals the first $n$
samples of $(N_{k,i},A_{k,i})$, so the uniform-in-$n$ bound applies under adaptive
sampling after a union bound over $k$ and $n$.
\end{lemma}

\begin{proof}
Fix $k$ and $n$, and write
$\widehat\mu_N=\frac1n\sum_{i=1}^n N_{k,i}$,
$\widehat\mu_A=\frac1n\sum_{i=1}^n A_{k,i}$.
Hoeffding's inequality yields, for any $\delta'>0$, with probability at least $1-\delta'$,
\begin{equation}
\begin{aligned}
    |\widehat\mu_N-\mu_N(k)|&\le N_{\max}\sqrt{\frac{\log(2/\delta')}{2n}},\\
    |\widehat\mu_A-\mu_A(k)|&\le A_{\max}\sqrt{\frac{\log(2/\delta')}{2n}}.
\end{aligned}
\end{equation}
Since $A_{k,i}\ge B_{\min}$ almost surely, also $\widehat\mu_A\ge B_{\min}$ and $\mu_A(k)\ge B_{\min}$.
Therefore
\begin{equation}
\begin{aligned}
    \left|\frac{\widehat\mu_N}{\widehat\mu_A}-\frac{\mu_N(k)}{\mu_A(k)}\right|
    &=\left|\frac{\widehat\mu_N\mu_A(k)-\mu_N(k)\widehat\mu_A}{\widehat\mu_A\mu_A(k)}\right|\\
    &\le \frac{|\widehat\mu_N-\mu_N(k)|}{B_{\min}}
    +\frac{\mu_N(k)\,|\widehat\mu_A-\mu_A(k)|}{B_{\min}^2}\\
    &\le \frac{|\widehat\mu_N-\mu_N(k)|}{B_{\min}}
    +\frac{N_{\max}\,|\widehat\mu_A-\mu_A(k)|}{B_{\min}^2}.
\end{aligned}
\end{equation}
Combining with the two Hoeffding bounds and absorbing constants into $c$ yields
$|\widehat C_{k,n}-C(k)|\le c L_{\max}\sqrt{\log(1/\delta')/n}$.
Taking a union bound over $k\in\{1,\ldots,K_{\max}\}$ and $n\in\{1,\ldots,T\}$ with
$\delta'=\delta/(2K_{\max}T)$ gives~\eqref{eq:ratio_conc} after adjusting constants.
\end{proof}

\begin{theorem}[Regret Bound of UCB-SpecStop]
\label{thm:regret}
Assume Assumption~\ref{asm:bounded} and let $c>0$ be the universal constant in
Lemma~\ref{lem:ratio_concentration}. Run Algorithm~\ref{alg:ucb} with
\begin{equation}
\label{eq:ucb_index_scale}
    \beta \ge c,
    \qquad
    w_k(t)\triangleq L_{\max}\sqrt{\frac{\log(4K_{\max}T^2)}{T_k(t)}}.
\end{equation}
Let $\Delta_k=C(k)-C(k^\star)$ for $k^\star\in\arg\min_k C(k)$.
Then the expected regret satisfies
\begin{equation}
\label{eq:regret_gap_dep}
    \mathbb{E}[R(T)]
    =O\!\left(
    \sum_{k:\Delta_k>0}
    \frac{L_{\max}^2\log(K_{\max}T)}{\Delta_k}
    \right)
\end{equation}
and
\begin{equation}
\label{eq:regret_gap_free}
    \mathbb{E}[R(T)]
    =O\!\left(
    L_{\max}\sqrt{K_{\max}T\log(K_{\max}T)}
    \right).
\end{equation}
\end{theorem}

\begin{proof}
Apply Lemma~\ref{lem:ratio_concentration} with $\delta=T^{-2}$. Let $\mathcal{E}$
be the event that for all arms $k$ and all $1\le n\le T$,
$|\widehat C_{k,n}-C(k)|\le \beta L_{\max}\sqrt{\log(4K_{\max}T^2)/n}$.
Then $\mathbb{P}(\mathcal{E})\ge 1-O(T^{-2})$.

On $\mathcal{E}$, whenever a suboptimal $k$ is chosen at time $t$ with $T_k(t)\ge 1$,
the index ordering gives
$\widehat C_{k,t}-\beta w_k(t)\le \widehat C_{k^\star,t}-\beta w_{k^\star}(t)$,
where $\widehat C_{k,t}=S_N(k)/S_A(k)$ after $T_k(t)$ pulls.
Using the concentration bounds on both sides yields
$C(k)-2\beta w_k(t)\le C(k^\star)$, hence
$\Delta_k\le 2\beta w_k(t)$ and therefore
$T_k(t)\le 4\beta^2L_{\max}^2\log(4K_{\max}T^2)/\Delta_k^2$.
Summing regret contributions $\Delta_k T_k(T)$ over $k$ yields the gap-dependent
bound on $\mathcal{E}$.

Per-round regret is $O(N_{\max}/B_{\min})$ while $\mathbb{P}(\mathcal{E}^c)=O(T^{-2})$,
so the failure event contributes only $O(1)$ to $\mathbb{E}[R(T)]$.
This establishes the gap-dependent expected regret bound.

For the gap-free bound, fix $\varepsilon>0$.
Arms with $\Delta_k\le\varepsilon$ contribute at most $T\varepsilon$;
arms with $\Delta_k>\varepsilon$ contribute at most
$O\bigl(K_{\max}L_{\max}^2\log(K_{\max}T)/\varepsilon\bigr)$ on $\mathcal{E}$.
Choosing
$\varepsilon=L_{\max}\sqrt{K_{\max}\log(K_{\max}T)/T}$
balances the two terms and yields the stated
$O\bigl(L_{\max}\sqrt{K_{\max}T\log(K_{\max}T)}\bigr)$ bound on $\mathbb{E}[R(T)]$.
\end{proof}

The gap-free bound carries a leading
constant proportional to $L_{\max}$ (hence to $N_{\max}/B_{\min}$ up to
an additional $N_{\max}A_{\max}/B_{\min}^2$ term from denominator fluctuations).
This scale captures the
``signal-to-noise'' difficulty of the problem: large $N_{\max}$
(high-delay environments) increases cost variance, while small
$B_{\min}$ (low acceptance) reduces the information gained per round.
For typical ranges ($K_{\max} \leq 20$, $D_{\max}$ on the order of
tens to hundreds of ms), $N_{\max}$ and $L_{\max}$ remain moderate.

The dependence on $K_{\max}$ reflects the cost of exploring all arms.
However, Theorem~\ref{thm:phase_transition} implies that only
$O(\log d / \log(1/\alpha))$ arms are competitive in high-delay
regimes, and the remaining arms have large gaps $\Delta_k$ and are
quickly eliminated. This structural insight opens the door to tighter
instance-dependent analyses that exploit the unimodal structure of
$C(k,d)$, potentially reducing the effective arm count from $K_{\max}$
to $O(\log d)$.

\subsection{Contextual Extension}
\label{subsec:contextual}

When the network state $s_t$ is observable (e.g., through RTT
measurements or congestion indicators), the agent can leverage this
side information to accelerate learning. The key idea is simple:
maintain independent statistics per state, effectively running
$|\mathcal{S}|$ parallel bandit instances. Proposition~\ref{thm:markov}
motivates maintaining state-dependent statistics, since the optimal arm
may vary with the network state, and sharing statistics across states
would conflate different per-state objectives.

We maintain independent $(S_N, S_A, T)$ statistics per $(k,s)$ and run
the same lower-index rule within the active state:

\begin{algorithm}[t]
\caption{Contextual UCB-SpecStop}
\label{alg:contextual}
\begin{algorithmic}[1]
\REQUIRE State space $\mathcal{S}$, arms $\mathcal{K}$, horizon $T$,
confidence parameter $\beta>0$ (Corollary~\ref{cor:contextual_regret} assumes
$\beta\ge c$ with $c$ from Lemma~\ref{lem:ratio_concentration})
\ENSURE State-dependent optimal policy $\hat{k}^*(s)$ for all $s \in \mathcal{S}$
\STATE \textbf{Initialize:} $S_N(k,s) \leftarrow 0$, 
$S_A(k,s) \leftarrow 0$, $T_{k,s} \leftarrow 0$ for all $(k,s)$
\FOR{$t = 1, 2, \ldots, T$}
    \STATE Observe $s_t$
    \IF{some $k$ has $T_{k,s_t} = 0$}
        \STATE $k_t \leftarrow$ any such $k$
    \ELSE
        \STATE $k_t \leftarrow \arg\min_{k} 
        \bigl[S_N(k,s_t)/S_A(k,s_t) - 
        \beta\, L_{\max}\sqrt{\log(4|\mathcal{S}|K_{\max}T^2)/T_{k,s_t}}\bigr]$
    \ENDIF
    \STATE Play $k_t$; observe $N_t, A_t$; update 
    $(S_N,S_A,T)_{k_t,s_t}$
\ENDFOR
\RETURN $\hat{k}^*(s) \leftarrow \arg\min_{k} S_N(k,s)/S_A(k,s)$, $\forall s \in \mathcal{S}$
\end{algorithmic}
\end{algorithm}

We assume the context sequence $\{s_t\}$ is exogenous: conditional on the past,
the law of future contexts does not depend on the arm choices
$\{k_u\}_{u\le t}$ beyond the realized current context $s_t$ (equivalently,
$s_{t+1},s_{t+2},\ldots$ are not shifted by switching $k_t$ while holding
$s_t$ fixed). This rules out feedback where long drafts materially alter future
congestion in a way not captured by the observed state; without it, per-state
bandit decompositions need not be valid.

\begin{corollary}[Contextual UCB-SpecStop]
\label{cor:contextual_regret}
Assume that at each round $t$, the context $s_t\in\mathcal{S}$ is observed before
choosing $k_t$. Conditional on $s_t=s$ and $k_t=k$, the observation $(N_t,A_t)$ is
drawn from a fixed distribution $P_{s,k}$, independently of the past.
Assume moreover $\beta\ge c$, where $c$ is the constant in
Lemma~\ref{lem:ratio_concentration} (as in Theorem~\ref{thm:regret}).
Define
\begin{equation}
\label{eq:contextual_Cs}
\begin{aligned}
    C_s(k)&\triangleq
    \frac{\mathbb{E}[N_t\mid s_t=s,k_t=k]}{\mathbb{E}[A_t\mid s_t=s,k_t=k]},\\
    k^\star(s)&\in\arg\min_{k\in\mathcal{K}} C_s(k),
\end{aligned}
\end{equation}
and $\Delta_{k,s}\triangleq C_s(k)-C_s(k^\star(s))$.
Then contextual UCB-SpecStop (Algorithm~\ref{alg:contextual}), run independently in
each state $s$, satisfies
\begin{equation}
\label{eq:contextual_regret_gap}
    \mathbb{E}[R(T)]
    =O\!\left(
    \sum_{s\in\mathcal{S}}
    \sum_{k:\Delta_{k,s}>0}
    \frac{L_{\max}^2\log(|\mathcal{S}|K_{\max}T)}{\Delta_{k,s}}
    \right)
\end{equation}
and
\begin{equation}
\label{eq:contextual_regret_free}
    \mathbb{E}[R(T)]
    =O\!\left(
    L_{\max}\sqrt{|\mathcal{S}|\,K_{\max}\,T\log(|\mathcal{S}|K_{\max}T)}
    \right),
\end{equation}
relative to playing $k^\star(s_t)$ every round.
\end{corollary}

\begin{proof}
For each $s$, restrict to rounds with $s_t=s$. Conditional on this subsequence,
the observations for each arm $k$ satisfy the same boundedness and i.i.d.\ pull
subsequence structure as in Lemma~\ref{lem:ratio_concentration}.
Apply Theorem~\ref{thm:regret} with horizon $T_s\le T$ and take a union bound over
$s\in\mathcal{S}$ (adjusting the logarithmic factor to $\log(|\mathcal{S}|K_{\max}T)$).
For the gap-free bound, apply the gap-free part of Theorem~\ref{thm:regret} per state
with pull counts $T_s$ satisfying $\sum_s T_s=T$, then use
$\sum_s \sqrt{T_s}\le \sqrt{|\mathcal{S}|\,T}$ (Cauchy--Schwarz).
\end{proof}

\begin{figure}
    \centering
    \includegraphics[width=1\columnwidth]{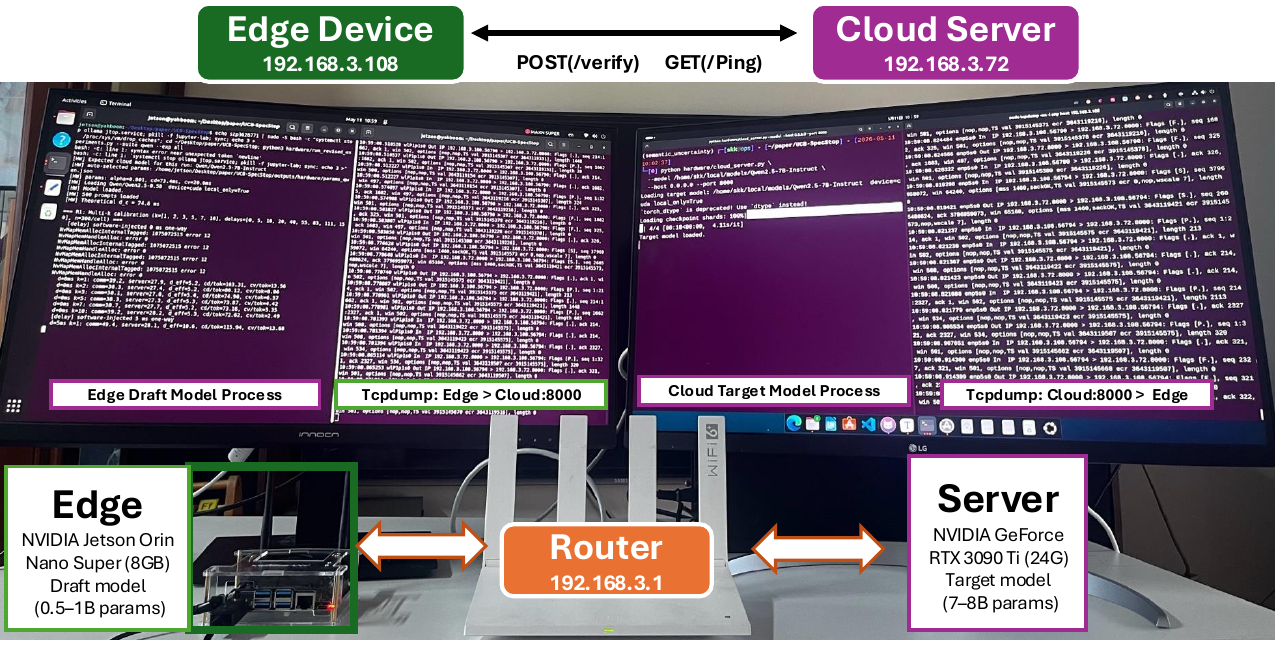}
    \caption{Experimental setup. The edge device (NVIDIA Jetson Orin Nano Super, 8 GB) hosts a draft model (0.5-1B params) and communicates with a cloud server (NVIDIA RTX 3090 Ti, 24 GB) running the target model (7-8B params) via a local router. The screens show (left to right) the edge draft process, outbound tcpdump trace, cloud verification process, and inbound tcpdump trace. The edge and cloud interact through POST(/verify) and GET(/Ping) HTTP endpoints.}
    \label{fig:exp_setup}
\end{figure}

In particular, if all positive gaps satisfy $\Delta_{k,s}\ge \Delta_{\min}>0$, then
\begin{equation}
\label{eq:contextual_regret_uniform_gap}
    \mathbb{E}[R(T)] =
    O\!\left(
    |\mathcal{S}|\,K_{\max}\,L_{\max}^2\log(|\mathcal{S}|K_{\max}T)/\Delta_{\min}
    \right).
\end{equation}

The multiplicative factor $|\mathcal{S}|$ reflects the cost of
learning a separate optimal arm in each state. In practice,
$|\mathcal{S}|$ is small (e.g., $|\mathcal{S}| = 2$ for a good/bad
channel model), so the overhead is modest. Furthermore, the structural monotonicity of the optimal stopping thresholds in
$s$ from Proposition~\ref{thm:markov} could be exploited to share information
across adjacent states, though we leave this refinement for future work.

Theorem~\ref{thm:voi} provides a principled criterion: if the delay
variance is small (all states yield the same $k^*$), the non-contextual
Algorithm~\ref{alg:ucb} suffices and converges faster (no state
splitting). If the variance is large enough for states to straddle the
phase transition $d_c$, contextual Algorithm~\ref{alg:contextual}
captures the VOI and achieves lower asymptotic cost. In practice, one
can start with the non-contextual version and switch to contextual when
observed delay variance exceeds a threshold derived from~\eqref{eq:dc}.

In summary, UCB-SpecStop addresses the practical challenge of unknown
parameters by combining a ratio-of-sums estimator (aligned with the
true objective~\eqref{eq:arm_cost}) with UCB-style exploration. The
$O\!\bigl(L_{\max}\sqrt{K_{\max} T \log(K_{\max} T)}\bigr)$ \emph{expected} regret guarantee ensures that the
cumulative cost of learning vanishes relative to the horizon $T$,
while the contextual extension leverages network state information to
match the performance of the oracle state-dependent policy from
Proposition~\ref{thm:markov}. This completes our answer to \textbf{Q3}.
The next section validates both the theoretical predictions and the
online algorithm through hardware measurements.

\section{Experimental Validation}\label{sec:experiments}

We validate the main theoretical predictions and UCB-SpecStop on a real
edge and cloud testbed using a six-round protocol R1--R6: per-arm cost
calibration (R1), empirical acceptance profiling (R2), phase-transition
sweeping (R3), fixed-delay strategy comparison (R4), online regret at
near-critical delay (R5, including running-cost convergence and
$\beta$-sensitivity analyses), and Markov-channel value of information (R6).
The subsections follow this pipeline so that measurement, modeling, and
learning claims are traceable to the same traces and configuration files.
The structural theorems are proved under constant per-token costs
$(c_d,c_v)$; in the testbed, batching induces mild $k$-dependent costs
(Table~\ref{tab:exp_calib}), so we report both an idealized theory oracle using
averaged costs and a calibrated oracle using per-$k$ costs.

\subsection{Testbed, Protocol, and Metrics}\label{subsec:exp_setup}

\begin{figure*}[t]
    \centering
    \subfloat[Qwen]{\includegraphics[width=0.48\textwidth]{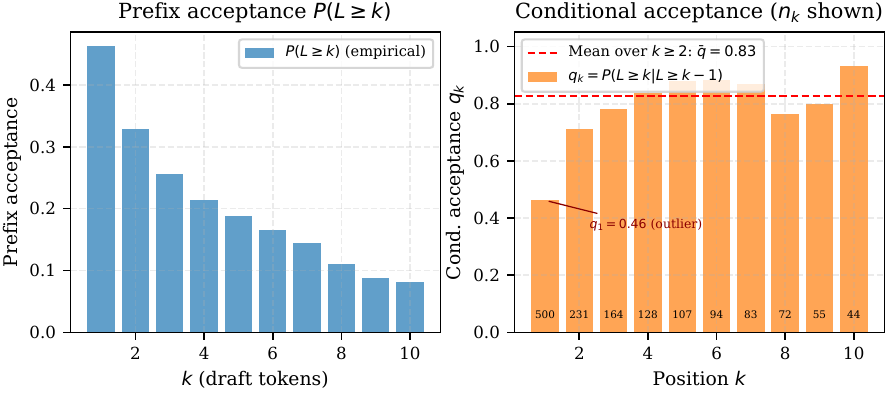}}
    \hfill
    \subfloat[LLaMA]{\includegraphics[width=0.48\textwidth]{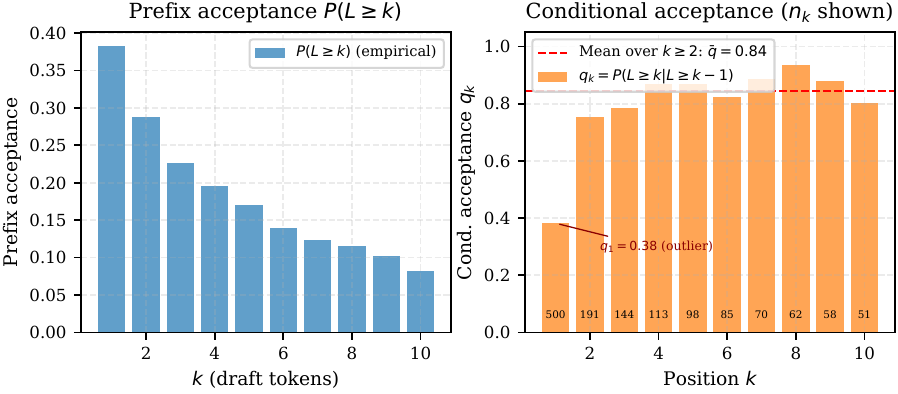}}
    \caption{Per-position acceptance. Left: $\hat q(k)=\Pr[L \ge k]$.
    Right: conditional $\Pr[L \ge k \mid L \ge k{-}1]$ and fitted $\alpha_\mathrm{geo}$.}
    \label{fig:exp_r2}
    \end{figure*}
\begin{table}[t]
    \centering
    \caption{Calibrated per-token costs (ms/token).}
    \label{tab:exp_calib}
    \resizebox{\columnwidth}{!}{%
    \setlength{\tabcolsep}{4pt}
    \footnotesize
    \begin{tabular}{@{}l r r r r r r r r r@{}}
    \toprule
    Suite & $\bar c_d$ & $\bar c_v$ & $c_d(k{=}1)$ & $c_d(k{=}5)$ & $c_d(k{=}10)$
    & $c_v(k{=}1)$ & $c_v(k{=}5)$ & $c_v(k{=}10)$ & RTT$_\mathrm{base}$ (ms)\\
    \midrule
    Qwen  & 85.14 & 8.09 & 106.25 & 79.46 & 73.70 & 16.56 & 5.50 & 3.06 & 10.01\\
    LLaMA & 67.37 & 8.37 & 90.40 & 58.94 & 52.59 & 17.18 & 5.78 & 3.12 & 9.02\\
    \bottomrule
    \end{tabular}%
    }
    \end{table}
    
    \begin{table}[t]
        \centering
        \caption{Per-position acceptance summary.}
        \label{tab:exp_accept}
        \begin{tabular}{@{}l r r r r r r@{}}
        \toprule
        Suite & $\hat q(1)$ & $\hat q(3)$ & $\hat q(5)$ & $\hat q(7)$ & $\hat q(10)$ & $\alpha_\mathrm{geo}$ ($k{\ge}2$)\\
        \midrule
        Qwen  & 0.462 & 0.256 & 0.188 & 0.144 & 0.082 & 0.828\\
        LLaMA & 0.382 & 0.226 & 0.170 & 0.124 & 0.082 & 0.845\\
        \bottomrule
        \end{tabular}
        \end{table}
The edge runs on an NVIDIA Jetson Orin Nano Super (8G) executing the
draft model; the cloud runs on an NVIDIA RTX~3090 Ti (24G) executing the target model.
The two nodes communicate over the LAN with controllable one-way delay injected at the cloud
interface via Linux \texttt{tc netem}. Bare-metal baseline RTTs measured during
calibration (Round~R1) are $10.01$\,ms (Qwen suite) and $9.02$\,ms (LLaMA
suite), consistent with a lightly loaded gigabit link.
We evaluate two draft and target pairs:
(i)~\textit{Qwen suite}, draft \texttt{Qwen/Qwen2.5-0.5B}, target
\texttt{Qwen/Qwen2.5-7B-Instruct};
(ii)~\textit{LLaMA suite}, draft \texttt{Llama-3.2-1B-Instruct}, target
\texttt{meta-llama/Llama-3.1-8B-Instruct}.
Per-token costs $c_d,c_v$ and acceptance profiles are calibrated in
\S\ref{subsec:exp_r1_r2}. The code is available at GitHub \footnote{https://github.com/szpsunkk/UCB-SpecStop}.

All cross-strategy comparisons use paired-prompt replay with deterministic
\texttt{verify\_seed = base + prompt\_id}. This design ensures prompt order and
verification randomness are aligned across strategies, so the measured gap is
driven only by strategy decisions.
Calibration artifacts are chained through a shared
\texttt{calibrated\_state.json}: R1 writes per-$k$ cost measurements,
R2 appends empirical acceptance curves, and R3 appends the per-delay empirical
oracle arm $\hat k^\star(d)$. Downstream rounds (R4 to R6) load these keys at
startup and warn on missing entries to avoid silent fallback to stale defaults.
Round budgets are: R3 uses 300 rounds per $(k,d)$ cell; R4 uses 1,000 rounds
per strategy per delay; R5 runs $T{=}5{,}000$ rounds at near-critical delay;
R6 runs 500 rounds under the Markov channel. 


We report the \emph{ratio-of-sums} per-token latency
$\widehat{C}=\sum_r T_r / \sum_r A_r$ (ms/token), aligned with the arm objective~\eqref{eq:arm_cost}
and the ratio-of-sums analysis in Section~\ref{sec:online}. Cumulative regret $R(t)$ in R5 uses the
offline best-fixed-arm empirical oracle $C^\star(d)$, the minimum ratio-of-sums cost across all
fixed-$k$ arms at the same delay, as the reference.

\subsection{Cost Calibration and Acceptance Profiling}\label{subsec:exp_r1_r2}

Direct measurement calibrates the per-token draft cost $c_d$ and verify cost
$c_v$ at each arm $k \in \{1,2,3,5,7,10\}$. These calibrated constants are used
by the theory/calibrated oracles.
Table~\ref{tab:exp_calib} shows three practical patterns. First, $c_d$
decreases with $k$ in both suites due to batching amortization (Qwen:
$106.25\!\to\!73.70$; LLaMA: $90.40\!\to\!52.59$ ms/token from $k{=}1$ to
$k{=}10$). Second, $c_v$ decreases even faster with $k$ because parallel
verification shares attention computation across positions
(Qwen: $16.56\!\to\!3.06$; LLaMA: $17.18\!\to\!3.12$ ms/token). Third, the
bare-metal RTT values (10.01 ms for Qwen, 9.02 ms for LLaMA) are measured
without injected delay and serve as the fixed LAN baseline; all reported
$d$ values are added by \texttt{tc netem} on top of that baseline. In the main
text, B4 uses averaged $(\bar c_d,\bar c_v)$ while B5 uses per-$k$ calibrated
costs; their close predictions in \S\ref{subsec:exp_r3} indicate that average-cost
approximation is sufficient for first-order tuning.
All reported delays are injected one-way delays. For oracle computation, we use
the effective one-way delay measured from traces,
$d_{\mathrm{eff}}=(\text{comm\_round}-\text{server\_total})/2$, so serialization
and RPC overhead are included consistently in both analysis and evaluation.
We report three oracle variants throughout the experiments:
theory oracle (B4, geometric acceptance with averaged costs),
calibrated-geometric oracle (B5, per-$k$ costs with geometric acceptance),
and offline best-fixed-arm empirical oracle (B6, per-$k$ costs with empirical prefix curve
$\hat B(k)=1+\sum_{i=1}^{k}\hat q(i)$).
The close agreement between B4 and B5 in \S\ref{subsec:exp_r3} suggests that
batching-induced cost variation across $k$ is a second-order effect for arm
selection, even though it still affects absolute latency prediction. By contrast,
the gap between B5 and B6 quantifies the price of the geometric-acceptance
approximation (Assumption~\ref{asm:geometric}): this gap is largest at low delays
where the heavy head $\hat q(1)$ dominates $\hat B(k)$, and shrinks at high
delays where tail terms become the main factor in $\arg\min_k C(k,d)$.
Using $d_{\mathrm{eff}}$ is also essential near the phase boundary. 

Fig.~\ref{fig:exp_r2} shows prefix survival $\hat q(k)=\Pr[L \ge k]$ (left) and
conditional continuation $\Pr[L \ge k \mid L \ge k{-}1]$ (right).
The head is heavy ($\hat q(1)=0.462/0.382$ for Qwen/LLaMA), while for $k\ge 2$ the tail is
near-geometric with $\alpha_\mathrm{geo}=0.828/0.845$, so post-head decay is roughly
constant-ratio---consistent with the geometric-tail abstraction and slow growth of
$k^\star(d)$ past the phase transition.
Table~\ref{tab:exp_accept} lists anchors at $k\in\{1,3,5,7,10\}$ for the calibrated oracle.
A single global $\alpha$ suffices for theory, but deployment should use the full $\hat q(k)$
when building the offline oracle; hence B6 in \S\ref{subsec:exp_r4} beats the pure geometric
oracle on our delay grid.
Surviving prompts drop along the prefix chain
(500$\rightarrow$231$\rightarrow\cdots\rightarrow$44 for Qwen and
500$\rightarrow$191$\rightarrow\cdots\rightarrow$51 for LLaMA up to $k{=}10$), yet
$k{=}10$ still leaves enough samples (44/51) for a stable tail read.
Relative to a pure geometric model, the empirical prefix nudges the low-delay optimum to
slightly larger $k$; at high delay, communication dominates and geometric and empirical
oracles agree, as in \S\ref{subsec:exp_r3}.
Acceptance is measured with deterministic paired-prompt replay
(\texttt{verify\_seed = base + prompt\_id}), so R4 gaps reflect $k$ choice, not verifier noise.
Overall, heavy head plus near-geometric tail justifies Assumption~\ref{asm:geometric} while
motivating empirical-prefix oracles.

\subsection{Phase Transition and Cost Curves}\label{subsec:exp_r3}

\begin{figure}[t]
    \centering
    \subfloat[Qwen]{\includegraphics[width=0.22\textwidth]{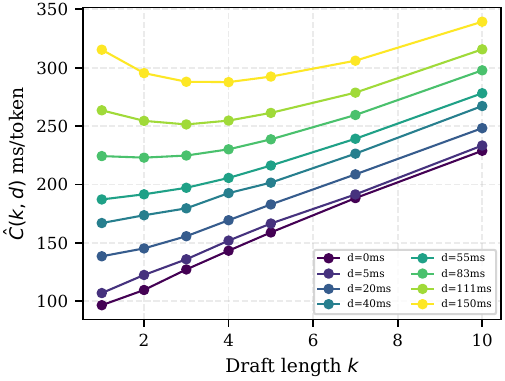}}
    \hfill
    \subfloat[LLaMA]{\includegraphics[width=0.22\textwidth]{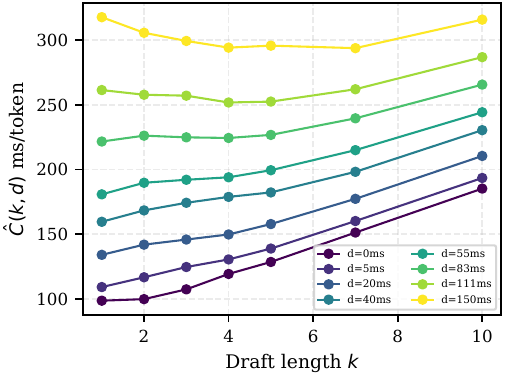}}
    \caption{Per-token cost $\widehat{C}(k,d)$ vs.\ $k$ for $d \in \{0,5,20,40,55,83,111,150\}$\,ms; minima highlighted.}
    \label{fig:exp_r3_cost}
    \end{figure}
    
    \begin{figure}[t]
    \centering
    \subfloat[Qwen]{\includegraphics[width=0.24\textwidth]{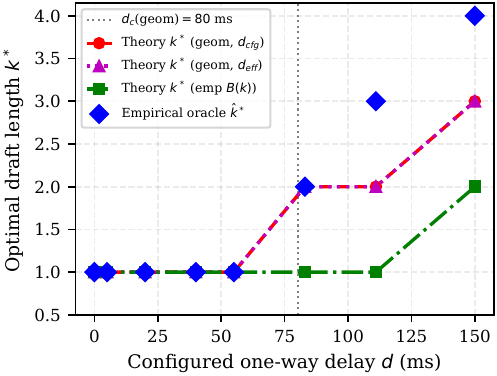}}
    \hfill
    \subfloat[LLaMA]{\includegraphics[width=0.24\textwidth]{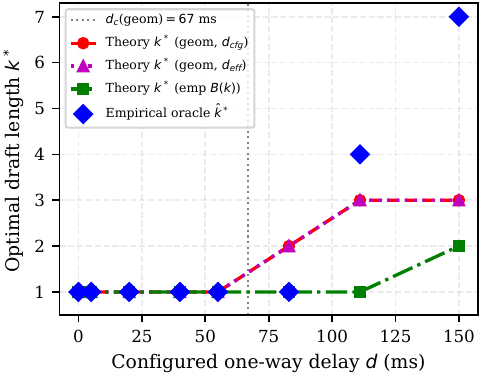}}
    \caption{Phase transition: empirical $\hat k^\star(d)$ (staircase) with geometric, calibrated-geometric, and empirical-prefix oracles.}
    \label{fig:exp_r3_phase}
    \end{figure}

Fig.~\ref{fig:exp_r3_cost} plots the measured per-token cost $\widehat{C}(k,d)$ as a function of the
draft length $k$ over the full one-way delay grid, and Fig.~\ref{fig:exp_r3_phase} reports
the empirical optimum $\hat k^\star(d)$ together with the three theoretical
predictions (geometric, calibrated-geometric, and empirical-prefix oracles).

Three observations match Theorem~\ref{thm:phase_transition} (phase transition and logarithmic scaling)
and the monotonicity result of Theorem~\ref{thm:monotonicity}.
(i)~The cost curves $\widehat{C}(\cdot,d)$ are U-shaped and strictly convex around their empirical minimum for every
$d$ in the grid, so a unique $\hat k^\star(d)$ is well defined.
(ii)~For both suites the sub-critical region $\hat k^\star(d)=1$ covers all small and
moderate delays, and the first jump to $\hat k^\star \ge 2$ occurs at a
suite-specific critical delay: $d_c^{(\mathrm{Qwen})} \approx 83$\,ms and
$d_c^{(\mathrm{LLaMA})} \approx 111$\,ms.
(iii)~Beyond the transition,
$\hat k^\star$ grows slowly: Qwen moves $1 \to 2 \to 3 \to 4$ between
$55$ and $150$\,ms; LLaMA moves $1 \to 4 \to 7$. Both trajectories are
consistent with the $\lceil \log d / \log(1/\alpha) \rceil$ envelope in
Theorem~\ref{thm:phase_transition}.

\begin{table}[t]
\centering
\caption{Hardware-measured phase transition. Bold marks the first delay with $\hat k^\star \ge 2$ (measured $d_c$).}
\label{tab:exp_phase}
\footnotesize
\resizebox{\columnwidth}{!}{%
\begin{tabular}{@{}l cccccccc@{}}
\toprule
$d$ (ms) & 0 & 5 & 20 & 40 & 55 & 83 & 111 & 150\\
\midrule
\textbf{Qwen} $\hat k^\star$ & 1 & 1 & 1 & 1 & 1 & \textbf{2} & 3 & 4\\
\textbf{Qwen} $\widehat{C}(\hat k^\star,d)$ (ms/tok) & 96.62 & 107.02 & 138.52 & 166.98 & 187.19 & 223.00 & 251.45 & 287.68\\
\textbf{Qwen} $k^\star_\mathrm{geom}$ (theory) & 1 & 1 & 1 & 1 & 1 & 2 & 2 & 3\\
\textbf{LLaMA} $\hat k^\star$ & 1 & 1 & 1 & 1 & 1 & 1 & \textbf{4} & 7\\
\textbf{LLaMA} $\widehat{C}(\hat k^\star,d)$ (ms/tok) & 98.73 & 109.23 & 134.12 & 159.65 & 180.82 & 221.63 & 251.76 & 293.72\\
\textbf{LLaMA} $k^\star_\mathrm{geom}$ (theory) & 1 & 1 & 1 & 1 & 1 & 2 & 3 & 3\\
\bottomrule
\end{tabular}}
\end{table}

These results provide three levels of evidence. First, the empirical optimum
is non-decreasing with injected delay, directly supporting
Theorem~\ref{thm:monotonicity}. Second, both suites exhibit a staircase-like
phase-transition pattern, supporting the qualitative structure of
Theorem~\ref{thm:phase_transition}. Third, exact transition locations are
calibration-dependent: the geometric model predicts Qwen well, while for LLaMA
it underestimates the optimal draft length at high delays; this discrepancy is
consistent with the heavy-head acceptance profile in Fig.~\ref{fig:exp_r2} and
is substantially corrected by the empirical-prefix oracle.

\begin{figure*}[t]
    \centering
    \subfloat[Qwen]{\includegraphics[width=0.48\textwidth]{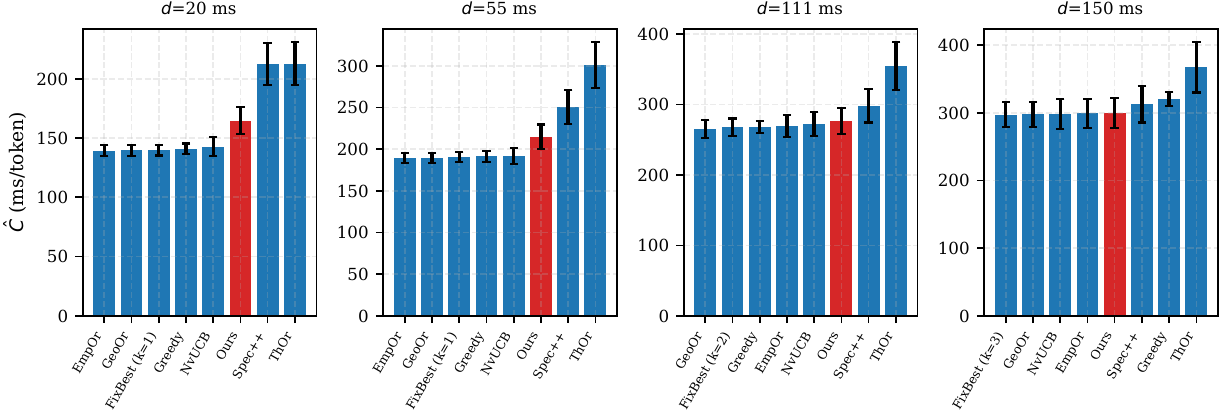}}
    \hfill
    \subfloat[LLaMA]{\includegraphics[width=0.48\textwidth]{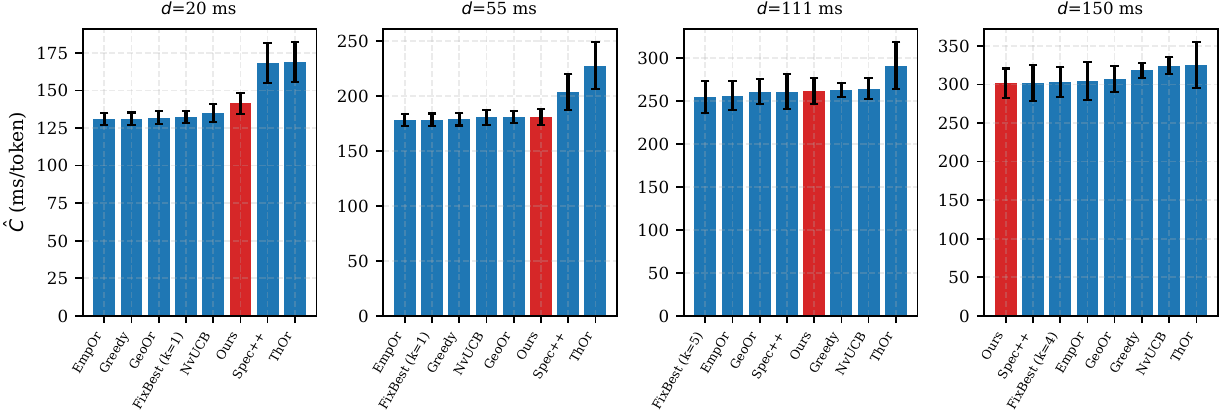}}
    \caption{Strategy comparison at four delays. Grouped bars by strategy; annotations mark the per-delay gap of our algorithm to the offline best-fixed-arm empirical oracle. EmpOr, GeoOr, ThOr, and NvUCB denote the best-fixed empirical oracle, calibrated-geometric oracle, theory oracle, and Naive-UCB baseline, respectively.}
    \label{fig:exp_r4}
    \end{figure*}

\subsection{Strategy Comparison at Fixed Delays}\label{subsec:exp_r4}

Fig.~\ref{fig:exp_r4} compares our algorithm against seven baselines at four delay points
spanning the sub-critical ($20$\,ms), near-critical ($55$\,ms), post-transition
($111$\,ms), and large-delay ($150$\,ms) regimes. The baselines are: (B1)
Fixed-$k$ for $k \in \{1,2,3,4,5,7,10\}$, we report the per-delay
best-fixed; (B2)~Greedy / zero-delay oracle; (B3)~SpecDec++ entropy-threshold
early-exit~\cite{huang2024specdecpp}; (B4)~theory-oracle $k^\star_\mathrm{geom}$; (B5)
calibrated-geometric oracle; (B6)~offline best-fixed-arm empirical oracle $\hat k^\star$; (B7)
Naive-UCB on the biased per-round mean estimator $\mathrm{mean}(T/A)$. Each strategy runs
$1{,}000$ rounds with identical prompt streams and \texttt{tc-netem} delay profiles.

The resulting per-delay picture yields four practical findings.
(1)~UCB-SpecStop matches the offline best-fixed-arm empirical oracle once $d$ approaches $d_c$.
For Qwen the UCB-vs.-empirical-oracle gap is $+13.5\%$ at $d{=}55$\,ms and
collapses to $+2.4\%$ at $111$\,ms and $+0.2\%$ at $150$\,ms. For LLaMA the
gap is $+1.5\%$ at $55$\,ms, $+2.1\%$ at $111$\,ms, and $-1.1\%$ (i.e.,
UCB beats the offline best-fixed-arm empirical oracle) at $150$\,ms by adapting within the run.
(2)~Naive-UCB (B7) incurs a visible tax in high-delay, higher-variance regimes. It trails UCB-SpecStop at
large $d$ in both suites (e.g., LLaMA at $150$\,ms: naive $324.3$ vs.\
ours $301.4$, a $+7.5\%$ overhead), consistent with the Jensen-bias argument
behind our ratio-of-sums estimator.
(3)~Fixed-$k$ is brittle. The best fixed arm at $d{=}20$\,ms ($k{=}1$) is
$14.0$ to $18.7$\% worse than the best fixed arm at $d{=}150$\,ms ($k{=}3$ or
$k{=}4$), so any \emph{a priori} single-$k$ pick misses the operating regime
by a wide margin under delay drift.
(4)~SpecDec++ (B3) loses whenever the channel is dominated by
communication. Its entropy threshold issues short drafts that waste
the network round trip; the gap to the offline best-fixed-arm empirical oracle is $+52\%$ / $+32\%$
for Qwen at $20$/$55$\,ms and $+29\%$ / $+14\%$ for LLaMA, narrowing only
once $d \ge 111$\,ms.

\def\rfouroracleqwen{figure/qwen/fig_r4_with_oracle.pdf}
\def\rfouroraclellama{figure/llama/fig_r4_with_oracle.pdf}
\IfFileExists{figure/qwen/fig_r4_with_oracle.pdf}{}{%
\def\rfouroracleqwen{figure/qwen/fig_r4_strategy_compare.pdf}}
\IfFileExists{figure/llama/fig_r4_with_oracle.pdf}{}{%
\def\rfouroraclellama{figure/llama/fig_r4_strategy_compare.pdf}}

\begin{table*}[t]
\centering
\caption{Per-token latency (ms/tok) by strategy and delay. Bold: lowest ratio-of-sums cost in each column. $\Delta$: our relative gap to the offline best-fixed-arm empirical oracle (negative $=$ better than oracle).}
\label{tab:exp_strategy}
\resizebox{0.9\textwidth}{!}{%
\footnotesize
\begin{tabular}{@{}l r r r r r r r r@{}}
\toprule
Strategy & Qwen 20 & Qwen 55 & Qwen 111 & Qwen 150 & LLaMA 20 & LLaMA 55 & LLaMA 111 & LLaMA 150\\
\midrule
Fixed-$k$ (best) & \textbf{139.84} (k1) & 190.32 (k1) & 267.38 (k2) & \textbf{297.00} (k3) & 132.27 (k1) & \textbf{178.50} (k1) & \textbf{254.41} (k5) & 302.84 (k4)\\
Fixed-$k{=}5$ & 195.40 & 234.11 & 286.80 & 304.48 & 158.37 & 198.79 & 254.41 & 305.17\\
Greedy (B2) & 140.83 & 191.34 & \textbf{268.02} & 320.74 & 131.18 & 178.60 & 263.28 & 318.18\\
SpecDec++ (B3) & 212.23 & 250.18 & 298.33 & 312.80 & 168.24 & 203.56 & 261.10 & \textbf{301.98}\\
Theory-oracle (B4) & 212.91 & 300.55 & 354.46 & 366.98 & 168.71 & 227.36 & 291.41 & 325.23\\
Calib.-oracle (B5) & 139.79 & 189.77 & 264.88 & 297.53 & 131.94 & 180.73 & 261.04 & 307.12\\
\textbf{Best-fixed empirical oracle (B6)} & \textbf{139.38} & \textbf{189.31} & \textbf{269.55} & 299.29 & \textbf{130.81} & 178.02 & 256.53 & 304.63\\
Naive-UCB (B7) & 142.77 & 191.67 & 272.29 & 298.07 & 135.05 & 180.50 & 264.59 & 324.28\\
\textbf{Ours (UCB-SpecStop)} & 164.78 & 214.78 & 276.04 & 299.92 & 141.46 & 180.76 & 261.94 & 301.40\\
$\Delta$ Ours vs.\ B6 & $+18.2$\% & $+13.5$\% & $+2.4$\% & $+0.2$\% & $+8.1$\% & $+1.5$\% & $+2.1$\% & $\mathbf{-1.1}$\%\\
\bottomrule
\end{tabular}%
}
\end{table*}

The $\Delta$ row shows that UCB-SpecStop pays an exploration tax in the
sub-critical region (small $d$) where every arm is close to optimal, but closes to $\le 2.5$\% and
occasionally overtakes the offline best-fixed-arm empirical oracle, past the phase
transition, which is where delay-aware draft-length selection has leverage.

\subsection{Online-Learning Regret}\label{subsec:exp_r5}

Fig.~\ref{fig:exp_r5_regret} plots cumulative regret $R(t)$ on a log scale for both axes for our algorithm
and two bandit baselines, Naive-UCB (biased $\mathrm{mean}(T/A)$) and EXP3
adapted to the ratio objective, at the near-critical delay for each suite
($d{=}83$\,ms Qwen, $d{=}111$\,ms LLaMA). The offline best-fixed-arm empirical oracle $\widehat{C}^\star(d)$ is computed
offline from the R3 cost grid as the minimum ratio-of-sums cost across all fixed-$k$ arms
at the same delay and used as the reference for instantaneous regret
$r_t=\widehat{C}_t-\widehat{C}^\star$. All three methods share the same prompt
stream and the same delay-shaped channel.

Before online learning starts, R5 also runs a same-delay per-arm oracle probe on
the same prompt stream to anchor fixed-$k$ references under identical channel
conditions. This probe is used only for plotting/diagnostics and does not alter
the regret definition above.

Two readings matter.
(i)~UCB-SpecStop and Naive-UCB both show slope close to $\tfrac{1}{2}$ in the two-axis log regret plot, empirically consistent with the gap-free
$O(\sqrt{T\log T})$ scaling of Theorem~\ref{thm:regret} (up to $\log K_{\max}$ factors), whereas EXP3 shows slope near $1$
and accrues about $1.4\times$ (Qwen) to $12\times$ (LLaMA) more regret.
(ii)~Ratio-of-sums vs.\ mean-of-ratios is the more subtle comparison:
UCB-SpecStop and Naive-UCB have similar empirical slopes in this near-critical
low-variance setting, and
their final-round difference is within the 95\% confidence band
($R_T{=}447{,}657$ vs.\ $448{,}805$\,ms for Qwen; $65{,}601$ vs.\ $65{,}651$\,ms for LLaMA).
We therefore interpret Fig.~\ref{fig:exp_r5_regret} as empirical consistency
with the gap-free trend rather than strict method dominance, and keep the estimator
ablation as the principled reason to prefer ratio-of-sums: it is the only
estimator for which the $O(\sqrt{T \log T})$ guarantee is known to hold
(Section~\ref{sec:online}), and it is the one whose performance advantage widens when the
per-arm variance is non-trivial (from the R4 strategy comparison in \S\ref{subsec:exp_r4}:
Qwen at $d{=}111$\,ms, $+1.0$\% naive vs.\
$+2.4$\% ours).

\begin{figure}[t]
\centering
\subfloat[Qwen]{\includegraphics[width=0.24\textwidth]{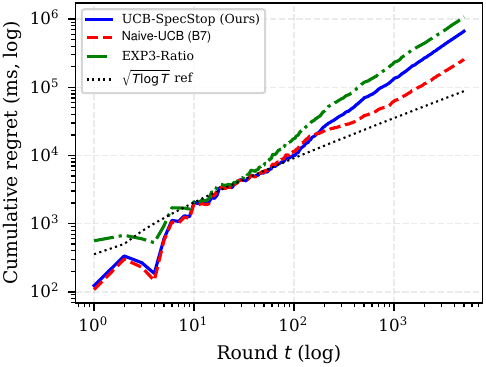}}
\hfill
\subfloat[LLaMA]{\includegraphics[width=0.24\textwidth]{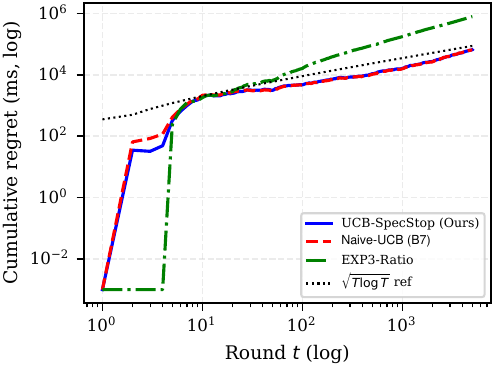}}
\caption{Cumulative regret with logarithmic scales on both axes. Shaded bands are 95\% CI across 30 bootstrap trajectories.}
\label{fig:exp_r5_regret}
\end{figure}

\begin{table}[t]
\centering
\caption{R5 final-round summary ($T{=}5{,}000$).}
\label{tab:exp_r5}
\resizebox{\columnwidth}{!}{%
\begin{tabular}{@{}l r r r r r r r@{}}
\toprule
Suite & $d$ (ms) & $\widehat{C}^\star$ & Ours $R_T$ & Naive $R_T$ & EXP3 $R_T$ & Ours $\widehat{C}$ & Gap\\
\midrule
Qwen  & 83  & 188.06 & 448{,}805 & 447{,}657 & 636{,}564 & 192.01 & $+2.10$\%\\
LLaMA & 111 & 191.43 & 65{,}651  & 65{,}601  & 792{,}963 & 183.00 & $\mathbf{-4.40}$\%\\
\bottomrule
\end{tabular}%
}
\end{table}

The negative oracle gap on LLaMA is measured against an offline best-fixed-arm
empirical oracle, not against a fully adaptive oracle. UCB-SpecStop can mix
arms online, and under finite-sample ratio-of-sums evaluation this dynamic
behavior may outperform any single fixed arm at large delay.
This effect is finite-sample and trace-dependent and does not imply that any
online policy systematically beats a fully adaptive oracle in the large-$T$ limit.




Complementing Fig.~\ref{fig:exp_r5_regret}, Fig.~\ref{fig:exp_r5_conv}
plots running per-token cost across rounds, and Table~\ref{tab:exp_beta}
summarizes offline $\beta$ sensitivity on Qwen R5 logs.
The UCB-SpecStop trajectories approach the empirical-oracle level while dashed
fixed-$k$ references from the same-delay per-arm probe (see \S\ref{subsec:exp_r5})
remain separated, highlighting that online adaptation drives long-run efficiency.
Table~\ref{tab:exp_beta} shows regret is nearly flat for $\beta\in[0.5,2.0]$ with
best mean near $\beta{=}1.5$ and overlapping confidence intervals, so the default
coefficient is not brittle.
During R5, UCB-SpecStop concentrates on the near-critical best arm after early
exploration; running cost therefore enters a near-oracle band by mid-horizon,
while EXP3 remains slower in this low-variance regime.
The LLaMA panel again illustrates that adaptive arm mixing can beat the best
\emph{fixed} arm under the ratio-of-sums objective, which reconciles the
negative oracle gap in Table~\ref{tab:exp_r5}.

\begin{figure}[t]
\centering
\subfloat[Qwen]{\includegraphics[width=0.24\textwidth]{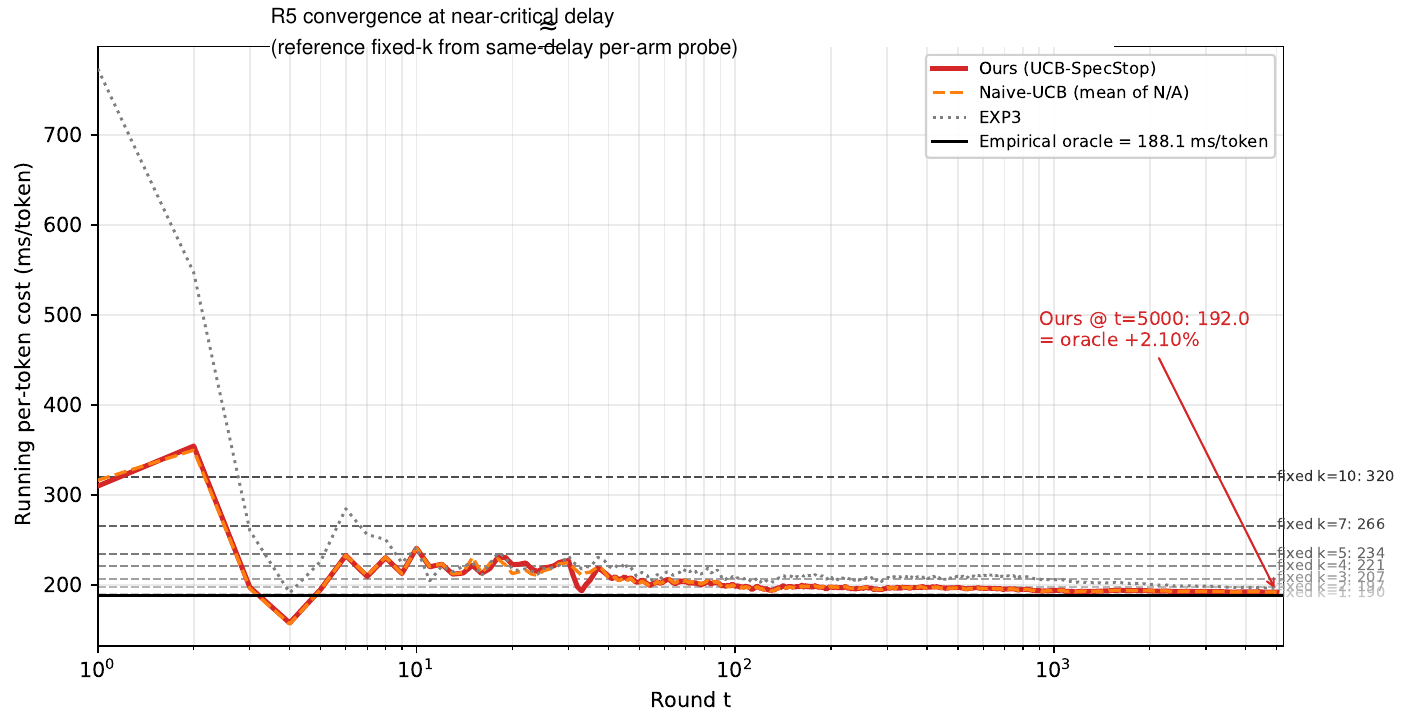}}
\hfill
\subfloat[LLaMA]{\includegraphics[width=0.24\textwidth]{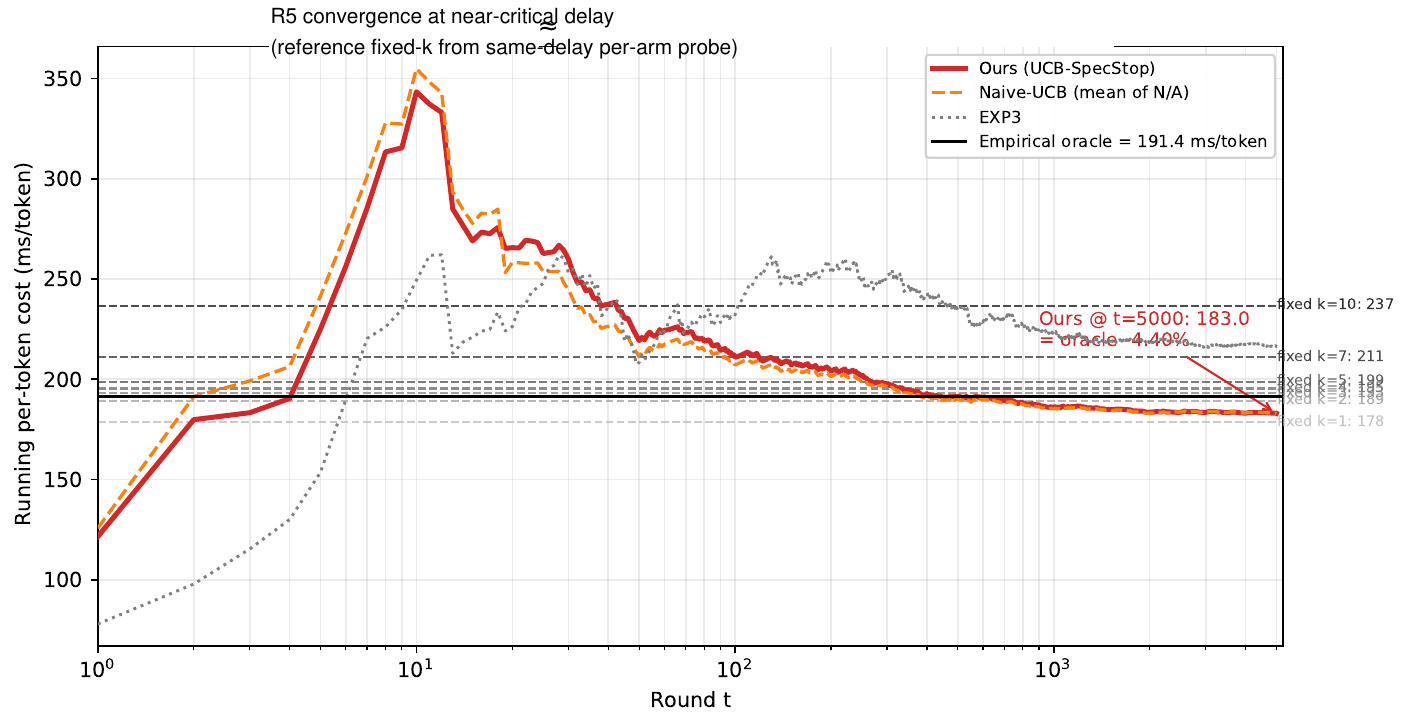}}
\caption{Running cost convergence. Solid black: offline best-fixed-arm empirical oracle.
Dashed greys: fixed-$k$ references from the R5 same-delay per-arm probe at
matched near-critical delays ($d{=}83$\,ms for Qwen, $d{=}111$\,ms for LLaMA).
Red: UCB-SpecStop.}
\label{fig:exp_r5_conv}
\end{figure}

\begin{table}[t]
\centering
\caption{Qwen R5: mean cumulative regret vs.\ $\beta$ (bootstrap).}
\label{tab:exp_beta}
\resizebox{\columnwidth}{!}{%
\footnotesize
\begin{tabular}{@{}l r r r r r r@{}}
\toprule
$\beta$ & 0.3 & 0.5 & 0.7 & 1.0 & 1.5 & 2.0\\
\midrule
$\bar R_T$ (ms) & 818{,}083 & 800{,}867 & 800{,}921 & 789{,}012 & \textbf{788{,}986} & 789{,}046\\
$\pm 95$\% CI & 118{,}825 & 116{,}211 & 116{,}234 & 114{,}960 & 114{,}946 & 114{,}919\\
\bottomrule
\end{tabular}%
}
\end{table}

\subsection{Value of Network-State Information}\label{subsec:exp_r6}

\begin{figure}[t]
    \centering
    \subfloat[Qwen]{\includegraphics[width=0.24\textwidth]{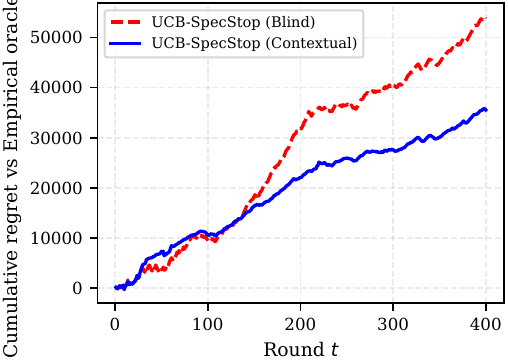}}
    \hfill
    \subfloat[LLaMA]{\includegraphics[width=0.24\textwidth]{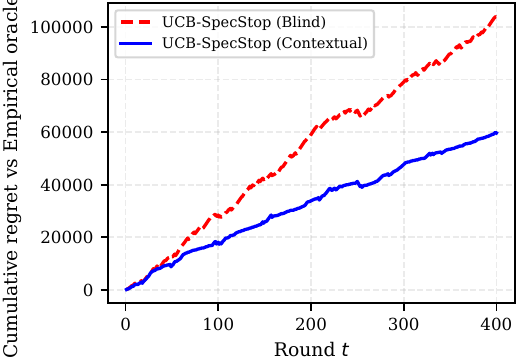}}
    \caption{VOI under Markov good/bad channel:
    contextual UCB-SpecStop vs.\ blind UCB-SpecStop.}
    \label{fig:exp_r6}
    \end{figure}

\begin{table}[t]
    \centering
    \caption{R6 VOI summary ($p_{g\to b}=p_{b\to g}=0.1$).}
    \label{tab:exp_voi}
    \resizebox{\columnwidth}{!}{%
    \footnotesize
    \begin{tabular}{@{}l l l r r r r@{}}
    \toprule
    Suite & $d_g$ / $d_b$ & $k^\star_g$ / $k^\star_b$ & $C_g$ / $C_b$ & Blind $\widehat{C}$ & Ctx.\ $\widehat{C}$ & VOI\\
    \midrule
    Qwen  & 37 / 111 & 1 / 2 & 97.34 / 170.53 & 228.72 & 221.82 & $\mathbf{+3.02}$\%\\
    LLaMA & 27 / 83  & 1 / 2 & 77.24 / 135.24 & 240.32 & 223.95 & $\mathbf{+6.81}$\%\\
    \bottomrule
    \end{tabular}%
    }
    \end{table}

Fig.~\ref{fig:exp_r6} evaluates Theorem~\ref{thm:voi} under a two-state Markov
channel. Contextual UCB-SpecStop outperforms blind UCB-SpecStop in both suites,
with positive VOI summarized in Table~\ref{tab:exp_voi}.
Each panel plots cumulative regret under the same good/bad channel realization,
comparing a context-aware policy (observes the state) to a blind policy
(state-agnostic). The consistent gap indicates that network-state awareness
reduces learning and control error when the channel alternates between
low-delay and high-delay regimes.
Table~\ref{tab:exp_voi} explains the gain: under good-state delay $d_g$ both
suites prefer shorter drafts, while under bad-state delay $d_b$ the optimum is
longer. A blind policy averages these regimes and pays mismatch cost, whereas the
contextual policy tracks state-specific arms, yielding
$+3.02\%$ (Qwen) and $+6.81\%$ (LLaMA) VOI.
The Markov channel uses symmetric transitions $p_{g\to b}=p_{b\to g}=0.1$
(stationary mass $(0.5,0.5)$, expected sojourn length 10 rounds per state).
The selected delay pairs place good and bad states on opposite sides of the
phase-transition neighborhood ($k_g^\star{=}1$, $k_b^\star{=}2$), matching the
regime where Theorem~\ref{thm:voi} predicts strictly positive VOI.
LLaMA exhibits larger VOI than Qwen because its cost curve is steeper around the
transition, so arm mismatch in the bad state is more expensive---a pattern
suggesting that RTT-aware adaptation matters most on high-variance links
(e.g., cellular or satellite paths).

\section{Conclusion}
\label{sec:conclusion}

Distributed LLM inference with edge drafts and cloud targets faces 
stochastic communication delays that change how speculation should be 
tuned. We cast draft-length selection as an optimal stopping problem and 
showed that the optimal policy under deterministic delay is a delay-monotone
threshold rule, that under a bounded speculation horizon $K_{\max}$ and monotone
stopping-region assumptions the Markov-modulated extension is a state-dependent
threshold, that the
optimal draft length scales only \textit{logarithmically} in delay with 
a sharp phase transition at a computable critical delay $d_c$, and that 
UCB-SpecStop with a ratio-of-sums estimator achieves a gap-dependent
logarithmic \emph{expected} regret bound and a gap-free
$O\!\bigl(L_{\max}\sqrt{K_{\max} T \log(K_{\max} T)}\bigr)$ \emph{expected} regret bound without prior
knowledge of the environment. The Markov-channel extension uses a Dinkelbach
transformation~\cite{dinkelbach1967} so that Bellman components share 
consistent units.
Section~\ref{sec:experiments} validates the core analytical predictions on a
Jetson Orin Nano Super and RTX~3090 testbed for Qwen 2.5 and
Llama~3 draft/target pairs using the full R1--R6 protocol and the
ratio-of-sums per-token metric. Measurements show suite-specific phase
transitions, agreement between geometric and empirical-prefix oracles after
calibration (R1--R3), strong comparisons against fixed-$k$ and heuristic
baselines (R4), sublinear regret for UCB-SpecStop at near-critical delay with
stable UCB tuning (R5), and strictly positive VOI for contextual policies on a
Markov good/bad channel (R6), aligning with
Theorems~\ref{thm:phase_transition},~\ref{thm:regret}, and~\ref{thm:voi}.
Together, these rounds underscore that deployment should calibrate costs and
acceptance from traces rather than from the idealized geometric model alone.




\bibliographystyle{IEEEtran}
\bibliography{bib}

\end{document}